\documentclass[acmsmall]{acmart}
\usepackage{booktabs}
\usepackage{hyperref}
\usepackage{cleveref}
\usepackage{algorithm}
\usepackage{algpseudocode}
\usepackage{subfig}
\algblock{Input}{EndInput}
\algnotext{EndInput}
\usepackage{enumitem}
\AtBeginDocument{%
  \providecommand\BibTeX{{%
    \normalfont B\kern-0.5em{\scshape i\kern-0.25em b}\kern-0.8em\TeX}}}

\acmJournal{JACM}
\acmVolume{37}
\acmNumber{4}
\acmArticle{111}
\acmMonth{8}

\begin{document}

\setcopyright{acmcopyright}
\acmJournal{TOIT}
\acmYear{2020} \acmVolume{1} \acmNumber{1} \acmArticle{1} \acmMonth{1} \acmPrice{15.00}\acmDOI{10.1145/3389249}

\title{A Blockchain-based Iterative Double Auction Protocol using Multiparty State Channels}

\author{Truc D. T. Nguyen}
\email{truc.nguyen@ufl.edu}
\orcid{0000-0002-5836-5884}
\affiliation{%
	\institution{University of Florida}
	\city{Gainesville}
	\state{Florida}
	\postcode{32611}
}
\author{My T. Thai}
\authornote{My T. Thai is the corresponding author.}
\email{mythai@cise.ufl.edu}
\affiliation{%
  \institution{University of Florida}
  \city{Gainesville}
  \state{Florida}
  \postcode{32611}
}

%
%
%
%
%
%


\begin{abstract}
	Although the iterative double auction has been widely used in many different applications, one of the major problems in its current implementations is that they rely on a trusted third party to handle the auction process. This imposes the risk of single point of failures, monopoly, and bribery. In this paper, we aim to tackle this problem by proposing a novel decentralized and trustless framework for iterative double auction based on blockchain. Our design adopts the smart contract and state channel technologies to enable a double auction process among parties that do not need to trust each other, while minimizing the blockchain transactions. In specific, we propose an extension to the original concept of state channels that can support multiparty computation. Then we provide a formal development of the proposed framework and prove the security of our design against adversaries. Finally, we develop a proof-of-concept implementation of our framework using Elixir and Solidity, on which we conduct various experiments to demonstrate its feasibility and practicality.

\end{abstract}



\maketitle

\section{Introduction}
Blockchain, the technology that underpins the great success of Bitcoin \cite{nakamoto2008bitcoin} and various other cryptocurrencies, has incredibly emerged as a trending research topic in both academic institutes and industries associations in recent years. With great potential and benefits, the blockchain technology promises a new decentralized platform for the economy such that the possibility of censorship, monopoly, and single point of failures can be eliminated \cite{swan2015blockchain}. The technology, in its simplest form, can be seen as a decentralized database or digital ledger that contains append-only data blocks where each block is comprised of valid transactions, timestamp and the cryptographic hash of the previous block. By design, a blockchain system is managed by nodes in a peer-to-peer network and operates efficiently in a decentralized fashion without the need of a central authority. Specifically, it enables a trustless network where participants of the system can settle transactions without having to trust each other. With the aid of the smart contracts technology, a blockchain system can enable a wide range of applications that go beyond financial transactions \cite{wood2014ethereum}. In the context of blockchain, \textit{smart contracts} are defined as self-executing and self-enforcing programs that are stored on chain. They are intended to facilitate and verify the execution of terms and conditions of a contract within the blockchain system. By employing this technology, applications that previously require a trusted intermediary can now operate in a decentralized manner while achieving the same functionality and certainty. For that reason, blockchain and smart contracts together have inspired many decentralized applications and stimulated scientific research in diverse domains \cite{nguyen2019optchain,saad2019partition,azaria2016medrec,dinh2018ai,kang2017enabling,aitzhan2016security,nguyen2018leveraging}. 


An auction is a market institution in which traders or parties submit bids that can be an offer to buy or sell at a given price \cite{friedman1993double}. A market can enable only buyers, only sellers, or both to make offers. In the latter case, it is referred as a two-sided or \textit{double auction}. A double auction process can be one-shot or iterative (repeated). The difference between them is that an iterative double auction process has multiple, instead of one, iterations \cite{parsons2006everything}. In each iteration, each party submits a bid illustrating the selling/buying price and supplying/demanding units of resource. This process goes on until the market reaches Nash Equilibrium (NE). In practice, the iterative double auction has been widely used for decentralized resource allocations among rational traders, especially for divisible resources, such as energy trading \cite{faqiry2016double,kang2017enabling}, mobile data offloading \cite{iosifidis2013iterative}, or resource allocation in autonomous networks \cite{iosifidis2010double}. In these applications, in each iteration, players submit their individual bids, respectively, to an auctioneer who later calculates to determine the resource allocation with respect to the submitted bids. However, current implementations of double auction systems require a centralized and trusted auctioneer to regulate the auction process. This results in the risk of single point of failures, monopoly, and bribery.

Although many research work have tried to develop a trading system combining the iterative double auction and blockchain \cite{kang2017enabling,wang2018decentralized}, nonetheless, they still need a trusted third-party to handle the auction process. In this work, we leverage blockchain and smart contracts to propose a general framework for iterative double auction that is completely \textit{decentralized} and \textit{trustless}. We argue that, due to the low throughput of blockchain, a naive and straightforward adoption of blockchain smart contracts to eliminate the trusted third-party would result in significantly high latency and transaction fees. To overcome this problem, we adopt the \textit{state channel} technology \cite{Dziembowski:2018:GSC:3243734.3243856} with extension to support computation among more than two parties, so that we can enable efficient off-chain execution of decentralized applications without changing the trust assumption. Specifically, we propose a double auction framework operating through a \textit{state channel} that can be coupled to existing double auction algorithms to run the auction process efficiently.

To demonstrate the feasibility and practicality of our solution, we develop a proof-of-concept implementation of the proposed solution. This proof-of-concept is built based on our novel development framework that can be used to deploy any distributed protocols using blockchain and state channels, which is also introduced in this paper. Based on the proof-of-concept, we conduct experiments and measure the performance in various aspects, the results suggest that our proposed solution can carry out a double auction process on blockchain that is both time- and cost-saving with a relatively small overhead.

\paragraph{Contributions.} We summarize our contributions as follows: 
\begin{itemize}
	\item We introduce a novel decentralized and trustless framework for iterative double auction based on blockchain. With this framework, existing double auction algorithms can efficiently run on a blockchain network without suffering the high latency and fees of on-chain transactions.
	\item We enhance the state channel technology, which is currently limited to two participants, to support multiparty computation. Based on this enhancement, we present a formal development of our solution, in which we develop a Universally Composable (UC)-style model \cite{canetti2001universally} for the double auction protocol and prove the security properties of our design using the simulation-based UC framework.
	\item To validate our proposed solution, we develop a proof-of-concept implementation of the framework to demonstrate its feasibility. For this implementation, we also introduce a novel development framework for distributed computing using blockchain and state channel. The framework is developed using the Elixir programming language and the system can be deployed on an Ethereum blockchain \cite{ethereum}.
\end{itemize}

\paragraph{Organization.} The remainder of this paper is organized as follows. The related work is summarized in Section \ref{sec:related}. Section \ref{sec:doubleauctionblockchain} discusses the integration of Blockchain and double auction to establish some security goals for the system. We first present a straw-man design and then provide a high-level view of our framework. Then, we provide formal security definitions and specifications with a detailed security analysis of our framework in Section \ref{sec:state}. Section \ref{sec:eval} presents the proof-of-concept implementation along with the system evaluation. Finally, Section \ref{sec:conclude} concludes the paper.

\section{Related work} \label{sec:related}
\paragraph{Double auction based on blockchain.}
As blockchain is an emerging technology, there has been many research work addressing double auction with blockchain. Recently, Thakur et al \cite{thakur2018distributed} published a paper on distributed double auction for peer to peer energy trading. The authors use the McAfee mechanism to process the double auction on smart contracts. In \cite{mingblockcloud}, the authors presented BlockCloud, which is a service-centric blockchain architecture that supports double auction. The auction model in this work uses a trade-reduction mechanism. However, the double auction mechanism in these work is one-shot and is only applicable to single-unit demands. For applications like energy or wireless spectrum allocation, these models greatly limit users' capability to utilize the products \cite{sun2014sprite}.

In \cite{kang2017enabling} and \cite{wang2018decentralized} , the authors propose blockchain-based energy trading using double auction. The auction mechanism is implemented as an iterative process which can be used for divisible goods. Although the system presented in these papers employs blockchain, the double auction process is still facilitated by a central entity. The blockchain is only used for settling payments. Our work is fundamentally different as we aim to design a framework that can regulate the iterative double auction process in a decentralized and trustless fashion.

\paragraph{State channel.} Although there has been many research effort on payment channels \cite{malavolta2017concurrency,miller2017sprites,dziembowski2019perun}, the concept of state channel has only emerged in recent years. As payment channel is limited to payment transactions, state channel is a generalization of payment channel in which users can execute complex smart contracts off-chain while still maintaining the trustless property. Instead of executing the contracts on-chain and having all the transactions validated by every blockchain nodes, the state channel technologies allow users to update the states off-chain with the option to raise disputes on-chain. Thus, it offers an efficient solution to address the scalability issue of blockchain systems. Dziembowski et al. \cite{Dziembowski:2018:GSC:3243734.3243856} is the first work that present formal specifications of state channels. However, the authors did not develop any proof-of-concept implementation to validate their protocol.

One problem with the original concept of state channels is that they only support execution between two parties, which is not applicable to our scenario since we are dealing with a system of multiple parties. For that reason, based on the work in \cite{Dziembowski:2018:GSC:3243734.3243856}, we extend the state channel technology to a multiparty state channel that can support computation among multiple users. Our extension also provides additional functionalities to handle dynamic changes of system participants. Based on this multiparty state channel, we design our double auction framework and analyze its security properties in the UC model.

\section{Double auction with blockchain} \label{sec:doubleauctionblockchain}
In this section, we formally define the double auction model that is used in this work. Moreover, beginning with a straw-man design, we present the high-level design of our framework and its security goals.
\subsection{Auction model}
We consider a set of parties that are connected to a blockchain network. We divide the set of parties into a set $\mathcal{B}$ of buyers who require resources from a set $\mathcal{S}$ of sellers. These two sets are disjoint. The demand of a buyer $i \in \mathcal{B}$ is denoted as $d_i$ and the supply of a seller $j \in \mathcal{S}$ is denoted as $s_j$. In this work, we adopt the auction model proposed in \cite{zou2016efficient}, which elicits hidden information about parties in order to maximize social welfare, as a general iterative double auction process that converges to a  Nash Equilibrium (NE).

A bid profile of a buyer $i \in \mathcal{B}$ is denoted as $b_i = (\beta_i, x_i)$ where $\beta_i$ is the buying price per unit of resource and $x_i$ is the amount of resource that $i$ wants to buy. Likewise, a bid profile of a seller $j \in \mathcal{S}$ is denoted as $b_j = (\alpha_j, y_j)$ where $\alpha_j$ is the selling price per unit of resource and $y_j$ is the amount of resource that $j$ wants to supply.

The auction process consists of multiple iterations. At an iteration $k$, the buyers and sellers submit their bid profiles $b^{(k)}_i$ and $b^{(k)}_j$, respectively, to the auctioneer. Then, a double auction algorithm will be used to determine the best response $b^{(k+1)}_i$ and $b^{(k+1)}_j$ for the next iteration. This process goes on until the auction reaches NE, at which the bid demand and supply $(x_i, y_j)$ will converge to an optimal value that maximizes the social welfare. An example of such algorithm can be referred to \cite{zou2016efficient}. The pseudo code for a centralized auctioneer is presented in \Cref{algo:sc}.

\begin{algorithm}[h]
	\caption{}\label{algo:sc}
	\begin{algorithmic}
		\State $k \gets 1$
		\While {not NE}
		\State Receive bid profiles $b^{(k)}_i$ and $b^{(k)}_j$ from buyers and sellers
		\State Compute best responses $b^{(k+1)}_i$ and $b^{(k+j)}_j$
		\State Send $b^{(k+1)}_i$ and $b^{(k+1)}_j$ back to sellers and buyers
		\State k $\gets$ k + 1
		\EndWhile
	\end{algorithmic}
\end{algorithm}

\subsection{Straw-man design}\label{sec:strawman}
In this section, we present a design of the trading system. The trading mechanism must meet the following requirements:
\begin{enumerate}[label=(\alph*)]
	\item Decentralized: the auction process is not facilitated by any central middleman
	\item Trustless: the parties do not have to trust each other.
	\item Non-Cancellation: parties may attempt to prematurely abort from the protocol to avoid payments. These malicious parties must be financially penalized.
\end{enumerate}

Based on the requirements, we will first show a straw-man design of the system which has some deficiencies in terms of latency and high transaction fee. Then, we will propose a trading system using state channels to address those problems.

In this system, we deploy a smart contract to the blockchain to regulate the trading process. Prior to placing any bid, all parties must make a deposit to the smart contract. If a party tries to cheat by prematurely aborting the trading process, he or she will lose that deposit and the remaining parties will receive compensation. Therefore, the deposit deters parties from cheating. At the end of the trading process, these deposits will be returned to the parties.

In this straw-man design, the auction process will be executed on-chain, that is, the smart contract will act as an auctioneer and thus will execute \Cref{algo:sc}. As the auction process consists of multiple iterations, the system will follow the activity diagram in \Cref{fig:paper} at each iteration.

\begin{figure}
	\centering
	\includegraphics[width=0.5\linewidth]{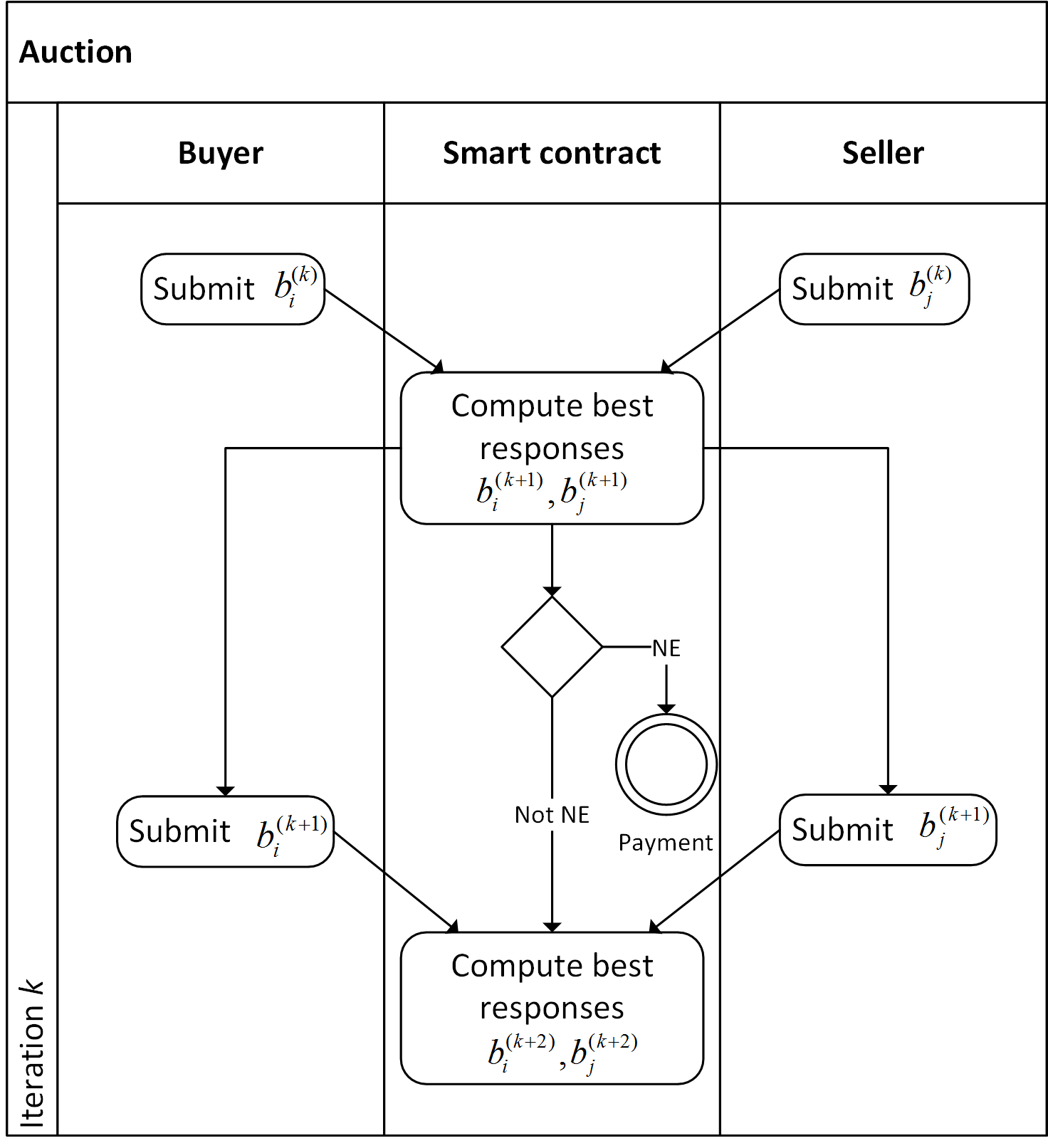}
	\caption{Auction phase}
	\label{fig:paper}
\end{figure}

At an iteration $k$, all buyers and sellers submit their bids $b^{(k)}_i$ and $b^{(k)}_j$, respectively, to the smart contract. In order to avoid unresponsiveness, a timeout is set for collecting bids. Should any parties fail to meet this deadline, the system considers that they aborted the process.

The smart contract then determines the best response $b^{(k+1)}_i$ and $b^{(k+j)}_j$ for buyers and sellers, respectively, until the trading system reaches NE. This design works, however, has two main disadvantages:
\begin{enumerate}
	\item Transaction latency: each message exchanged between parties and the smart contract is treated as a blockchain transaction which takes time to get committed.
	\item High computational complexity on the smart contract which means that the blockchain will require high transaction fees.
\end{enumerate}
These disadvantages come from the fact that a transaction in a blockchain network has to be confirmed by all the validators. In other words, with this design, the buyers and sellers are having the entire blockchain to process their auction, which is very inefficient and, thus, not practical. In the following section, we will propose another design to overcome these issues.

\subsection{Blockchain with multiparty state channels}
As the double auction process involves multiple iterations, state channel \cite{Dziembowski:2018:GSC:3243734.3243856} is a proper solution to address the deficiencies of the straw-man design. Instead of processing the auction on-chain, the parties will be able to update the states of the auction off-chain. Whenever something goes wrong (e.g., some parties try to cheat), the users always have the option of referring back to the blockchain for the certainty of on-chain transactions. Since the original concept of state channel only supports the computation between two parties, in this work, we propose an extension to support multiparty computation that can work with the double auction. This section illustrates a high-level view of how we can conduct a double auction using the multiparty state channel.

In the same manner as the straw-man design, the parties deploy a smart contract to the blockchain. However, this smart contract does not regulate the auction process, but instead acts as a judge to resolve disputes. The parties must also make a deposit to this contract prior to the auction. \Cref{fig:picture1} illustrates the overview of the operations in the multiparty state channel.

\begin{figure}
	\centering
	\includegraphics[width=0.8\linewidth]{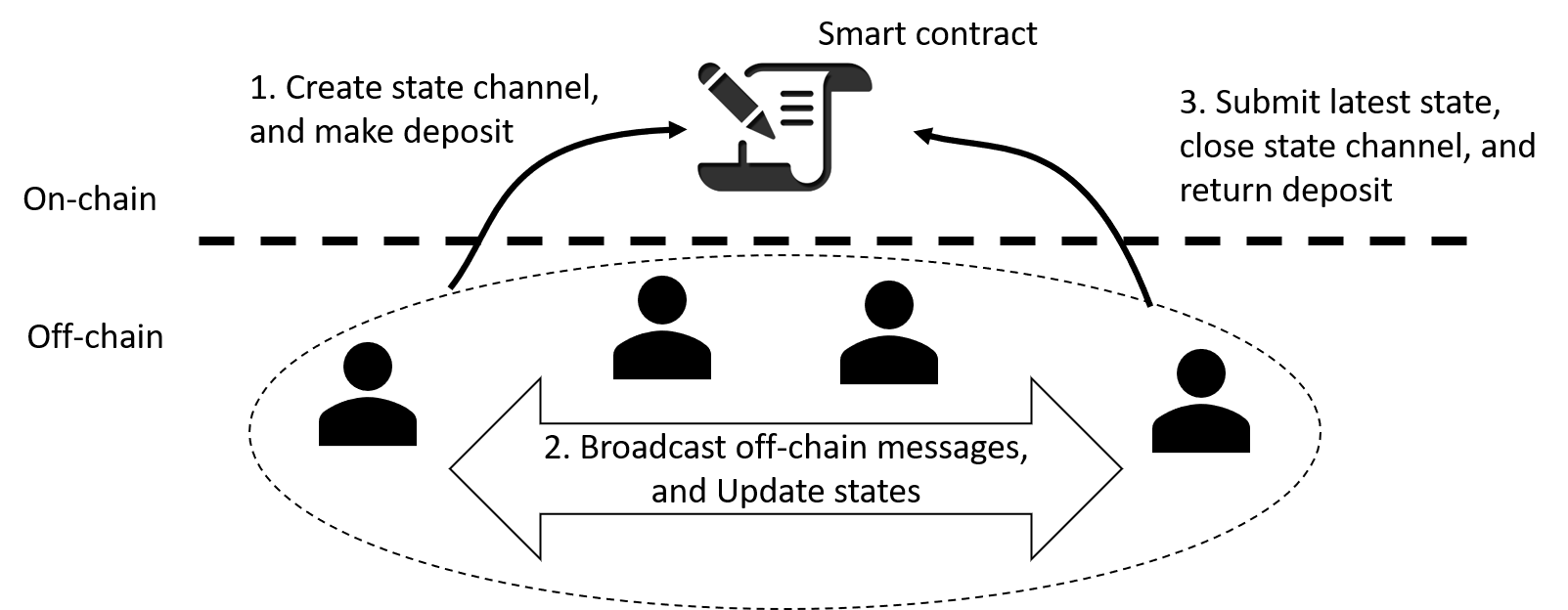}
	\caption{Multiparty state channel: overview}
	\label{fig:picture1}
\end{figure}

\begin{figure*}[ht!]
	\centering
	\includegraphics[width=0.9\linewidth]{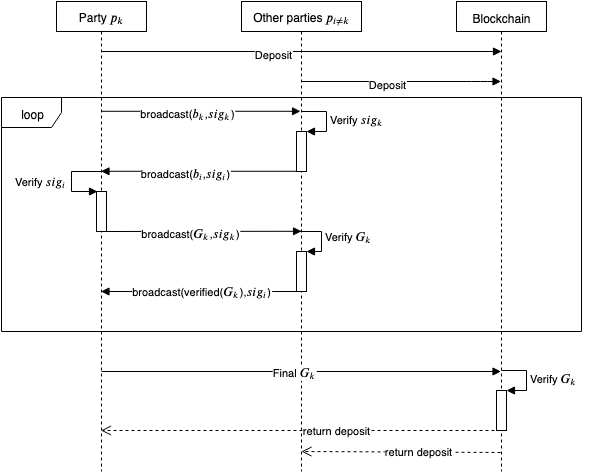}
	\caption{Sequence diagram of the double auction process. Here we single out one party $p_k$ to elucidate the interactions among parties}
	\label{fig:seq}
\end{figure*}

After deploying the smart contract, the parties can now begin the auction process in a \textit{state channel}. At each iteration $k$ of the auction, we define two operations: (1) collecting bids and (2) determining the best responses. Denoting the set of parties as $\mathcal{P} = \mathcal{B} \cup \mathcal{S} = \{p_1, p_2, ..., p_n\}$, in the first operation, each party broadcasts a blockchain transaction containing its bid $b^{(k)}_{p_i}$ to all other parties. Note that this transaction is a \textit{valid} blockchain transaction and it is only broadcasted locally among the parties. Upon receiving that transaction, each party has to verify the transaction's signature. After this operation, each party now has the bid profiles of all other parties. 

Then we move to the second operation of determining the best response. A party will be chosen to compute the best responses, in fact, it does not matter who will execute this computation because the results will later be verified by other parties. Therefore, this party can be chosen randomly or based on the amount of deposit to the smart contract. Let $p_k$ be the one who carries out the computation at iteration $k$, $G_k$ be the result that consists of the best response $b^{(k+1)}_{p_i}$ for each party $p_i$, $p_k$ will broadcast a blockchain transaction containing $G_k$ to all other parties. Upon receiving this transaction, each party has to verify the result $G_k$, then signs it and broadcasts another transaction containing $G_k$ to all other parties. This action means that the party agrees with $G_k$. After this step, each party will have $G_k$ together with the signatures of all parties.

When the auction process reaches NE, a party will send the final $G_k$ together with all the signatures to the smart contract. The smart contract then verifies the $G_k$ and if there is no dispute, the state channel is closed. Finally, the payment will be processed on-chain and the smart contract refunds the initial deposit to all parties. The entire process is summarized in the sequence diagram in \Cref{fig:seq}.

As can be seen, the blockchain is invoked only two times and thus saves tons of transaction fees comparing to the straw-man design. Moreover, as the transactions are not sent to the blockchain, the latency is only limited by the communication network among the parties. We can also see that the bid profiles are only known among the involving parties, not to the entire blockchain, thus enhances the privacy. In this section, we only provide a very abstract workflow of the system without considering the security and privacy. In the following sections, we provide more detail operations of the proposed protocol as well as how we ensure the system is secured.

\subsection{Security and privacy goals}
Before we present the formal development of our solution, we establish the threat model as well as security and privacy goals for our system. In this work, we consider a computationally efficient adversary who can corrupt any subset of parties. By corruption, the attacker can take full control over a party which includes acquiring the internal state and all the messages destined to that party. Moreover, it can send arbitrary messages on the corrupted party's behalf.

With respect to the adversarial model, we define the security and privacy notions of interest as follows:
\begin{itemize}
	\item Unforgeability: We use the ECDSA signature scheme which is believed to be unforgeable against chosen message attack \cite{johnson2001elliptic}. This signature scheme is currently being used by the Ethereum blockchain \cite{wood2014ethereum}.
	\item Non-Repudiation: Once a party has submitted a bid, the system assures that he or she must not be able to deny having made the relevant bid. 
	\item Public Verifiability: All parties can be publicly verified if they have been following the auction protocol correctly.
	\item Robustness: The auction process can tolerate invalid bids and dishonest participants who are not following the auction protocol correctly.
	\item Input independence: Each party does not see others' bid before committing to their own.
	\item Liveness: In an optimistic case when all parties are honest, the computation is processed within a small amount of time (off-chain messages only). When some parties are corrupted, the computation is completed within a predictable amount of time.
\end{itemize}


\section{Double auction using multiparty state channels}\label{sec:state}
In this section, we describe the ideal functionality of our system that defines how a double auction process is operated using the multiparty state channel technology. We show that the ideal functionality achieves the security goals. Afterwards, we present the design of our protocol that realizes the ideal functionality. Finally, a detailed security proof in the UC framework is given.

\subsection{Security model} \label{sec:model}
The entities in our system are modeled as interactive Turing machines that communicate with each other via a secure and authenticated channel. The system operates in the presence of an adversary $\mathcal{A}$ who, upon corruption of a party $p$, seizes the internal state of $p$ and all the incoming and outgoing packets of $p$.
\subsubsection{Assumptions and Notation}
We denote $\mathcal{P} = \mathcal{B} \cup \mathcal{S} = \{p_1, p_2, ..., p_n\}$ as the set of $n$ parties. We assume that $\mathcal{P}$ is known before opening the state channel and $|\mathcal{P}| \geq 2$. The blockchain is represented as an append-only ledger $\mathcal{L}$ that is managed by a global ideal functionality $\mathcal{F}_\mathcal{L}$ (such as \cite{Dziembowski:2018:GSC:3243734.3243856}). The state of $\mathcal{F}_\mathcal{L}$ is defined by the current balance of all accounts and smart contracts' state; and is publicly visible to all parties. $\mathcal{F}_\mathcal{L}$ supports the functionalities of adding or subtracting one's balance. We also denote $\mathcal{F}(x)$ as retrieving the current value of the state variable $x$ from an ideal functionality $\mathcal{F}$.

We further assume that any message destined to $\mathcal{F}_\mathcal{L}$ can be seen by all parties (in the same manner as blockchain transactions are publicly visible).  For simplicity, we assume that all parties have enough fund in their accounts for making deposits to the smart contract. Furthermore, each party and the ideal functionality will automatically discard any messages originated from a party that is not in $\mathcal{P}$ or the message's signature is invalid.

\subsubsection{Communication} 
In this work, we assume a synchronous communication network. We define a \textit{round} as a unit of time corresponding to the maximum delay needed to transmit an off-chain message between a pair of parties. Any modifications on $\mathcal{F}_\mathcal{L}$ and smart contracts take at most $\Delta \in \mathbb{N}$ rounds, this $\Delta$ reflects the fact that updates on the blockchain are not instant but can be completed within a predictable amount of time. Furthermore, each party can retrieve the current state of $\mathcal{F}_\mathcal{L}$ and smart contracts in one round.

\subsubsection{Commitment scheme}
One cryptographic primitive that is used in our model is the commitment scheme. A commitment scheme consists of two following algorithms $(Com, Vrf)$:
\begin{itemize}
    \item $Com(m, r)$: given a message $m$, random nonce $r$, returns commitment $c$
    \item $Vrf(c, R)$: given a commitment $c$, and $R \triangleq (m, r)$, where $m$ is a message, $r$ is a random nonce, returns $1$ iff $c = Com(m, r)$, otherwise returns $0$.
\end{itemize}
We assume that there is no adversary $\mathcal{A}$ that can generate a commitment $c$ and the tuples $(m, r)$, $(m', r')$, such that $c = Com(m, r) = Com(m', r')$. Simply speaking, a party cannot alter the value after they have committed to it.

\subsection{Ideal functionality}
First, we define the ledger's ideal functionality $\mathcal{F}_\mathcal{L}$. Based on \cref{sec:model}, the $\mathcal{F}_\mathcal{L}$ supports adding and subtracting one's balance, hence, we give the corresponding definition in \Cref{fig:fl}.

\begin{figure}[!htbp]
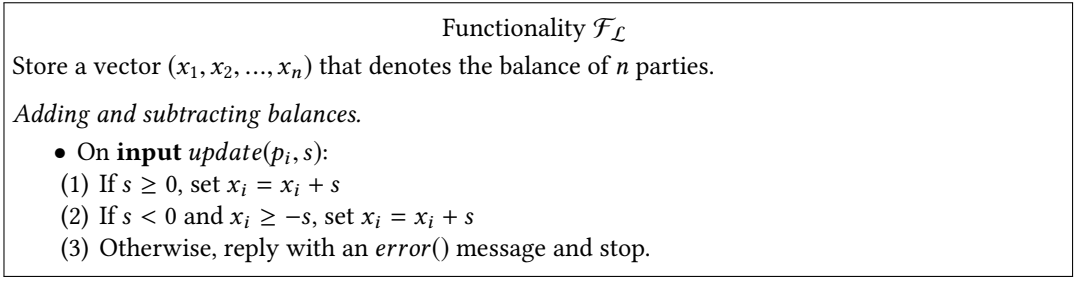

	\noindent\fbox{%
		\parbox{\columnwidth}{%
			\begin{center}
				Functionality $\mathcal{F}_\mathcal{L}$
			\end{center}
			
			Store a vector $(x_1, x_2, ..., x_n)$ that denotes the balance of $n$ parties.
			
			\paragraph{Adding and subtracting balances}
			\begin{itemize}
				\item On \textbf{input} $update(p_i, s)$:
				\begin{enumerate}
					\item If $s \geq 0$, set $x_i = x_i + s$
					\item If $s < 0$ and $x_i \geq -s$, set $x_i = x_i + s$
					\item Otherwise, reply with an $error()$ message and stop.
				\end{enumerate}
			\end{itemize}
		}
	}
	\caption{Ledger's functionality $\mathcal{F}_\mathcal{L}$}
	\label{fig:fl}
\end{figure}

The formal definition of the ideal functionality $\mathcal{F}_{auction}$ is presented in \Cref{fig:ideal}. As can be seen, it supports the following functionalities:
\begin{itemize}
	\item Open channel
	\item Determine best response
	\item Revocation
	\item Close channel
\end{itemize}

\begin{figure}
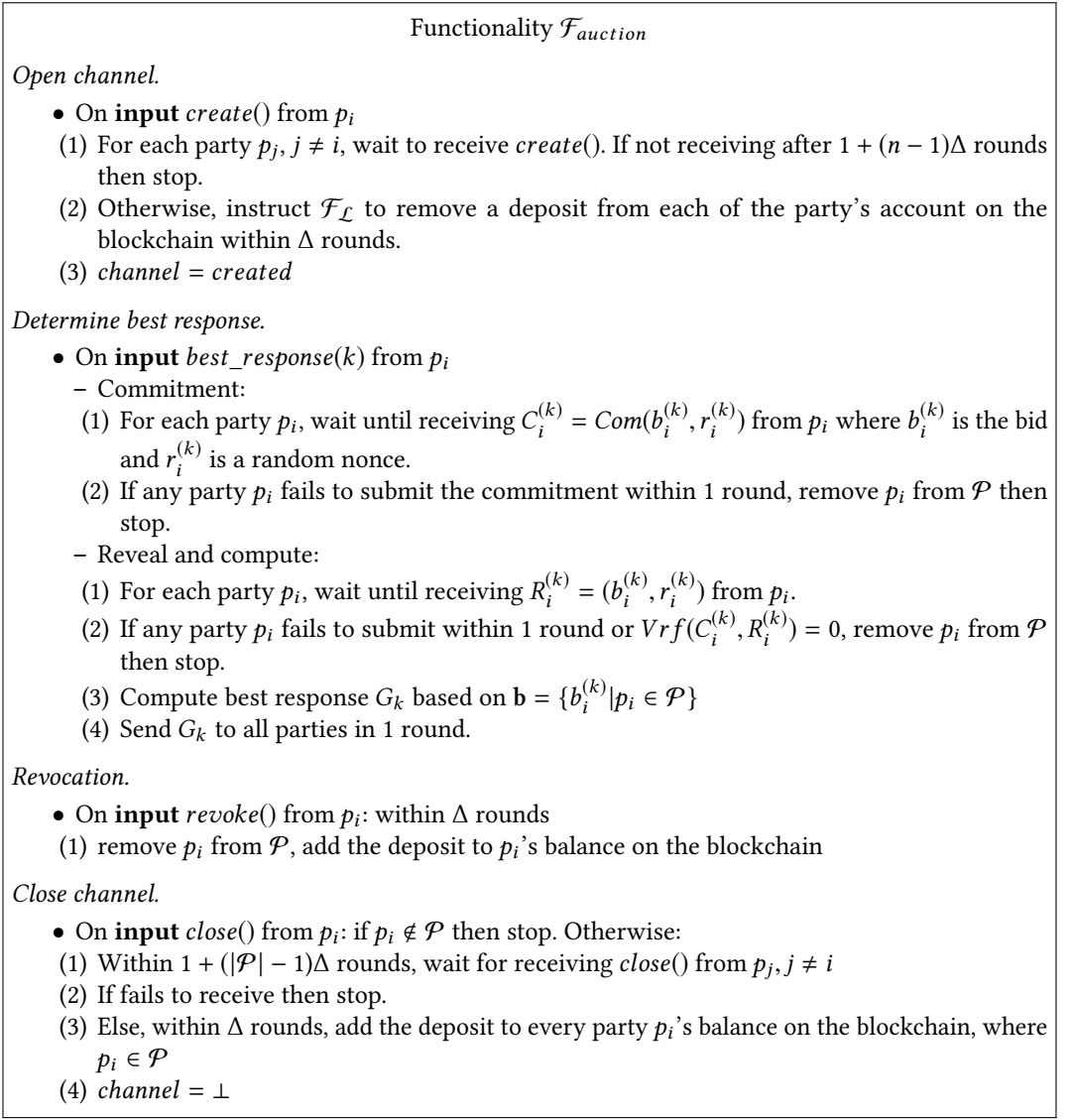

	\noindent\fbox{%
		\parbox{\columnwidth}{%
			\begin{center}
				Functionality $\mathcal{F}_{auction}$
			\end{center}
			
			\paragraph{Open channel}
			\begin{itemize}
				\item On \textbf{input} $create()$ from $p_i$
				\begin{enumerate}
					\item For each party $p_j$, $j \neq i$, wait to receive $create()$. If not receiving after $1+(n-1)\Delta$ rounds then stop.
					\item Otherwise, instruct $\mathcal{F}_\mathcal{L}$ to remove a deposit from each of the party's account on the blockchain within $\Delta$ rounds.
					\item $channel=created$
				\end{enumerate}
			\end{itemize}
			
			\paragraph{Determine best response}
			\begin{itemize}
				\item On \textbf{input} $best\_response(k)$ from $p_i$
				\begin{itemize}
					\item Commitment: 
					\begin{enumerate}
						\item For each party $p_i$, wait until receiving $C_i^{(k)} = Com(b_i^{(k)}, r_i^{(k)})$ from $p_i$ where $b_i^{(k)}$ is the bid and $r_i^{(k)}$ is a random nonce.
						\item If any party $p_i$ fails to submit the commitment within 1 round, remove $p_i$ from $\mathcal{P}$ then stop.
					\end{enumerate}
					\item Reveal and compute:
					\begin{enumerate}
						\item For each party $p_i$, wait until receiving $R_i^{(k)} = (b_i^{(k)}, r_i^{(k)})$ from $p_i$. 
						\item If any party $p_i$ fails to submit within 1 round or $Vrf(C_i^{(k)}, R_i^{(k)}) = 0$, remove $p_i$ from $\mathcal{P}$ then stop.
						\item Compute best response $G_k$ based on $\mathbf{b} = \{b^{(k)}_i | p_i \in \mathcal{P}\}$
						\item Send $G_k$ to all parties in 1 round.
					\end{enumerate}
				\end{itemize}
			\end{itemize}
			
			\paragraph{Revocation}
			\begin{itemize}
				\item On \textbf{input} $revoke()$ from $p_i$: within $\Delta$ rounds 
				\begin{enumerate}
					\item remove $p_i$ from $\mathcal{P}$, add the deposit to $p_i$'s balance on the blockchain
				\end{enumerate}
			\end{itemize}
			
			\paragraph{Close channel}
			\begin{itemize}
				\item On \textbf{input} $close()$ from $p_i$: if $p_i \notin \mathcal{P}$ then stop. Otherwise:
				\begin{enumerate}
					\item Within $1+(|\mathcal{P}| - 1$)$\Delta$ rounds, wait for receiving $close()$ from $p_j, j\neq i$
					\item If fails to receive then stop.
					\item Else, within $\Delta$ rounds, add the deposit to every party $p_i$'s balance on the blockchain, where $p_i \in \mathcal{P}$
					\item $channel=\bot$
				\end{enumerate}
			\end{itemize}
		}%
	}
	\caption{Ideal functionality $\mathcal{F}_{auction}$}
	\label{fig:ideal}
\end{figure}

As indicated in \cite{Dziembowski:2018:GSC:3243734.3243856}, a state channel should be able to guarantee the consensus on creation and closing. The state channel creation is initiated by receiving a $create()$ message from a party. The functionality then waits for receiving $create()$ from all other parties within $1+ (n-1)\Delta$ rounds. If this happens then the functionality removes a deposit from each party's balance on the blockchain. Since all parties have to send the $create()$ message, we achieve the consensus on creation.

Each iteration $k$ of the double auction process starts with receiving the $best\_response(k)$ message from a party. Then all parties must submit a commitment of their bids. After that, all parties must submit the true bid that matches with the commitment they sent before. Any party fails to submit in time or does not submit the true bid will be eliminated from the double auction process. With the commitment step, one party cannot see the other parties' bid prior to placing his or her own bid, this satisfies the \textit{Input independence}.

When one party fails to behave honestly, it will be eliminated from the auction process and will not receive the deposit back. A party can voluntarily abort an auction process by sending a $revoke()$ message and it will receive the deposit back.  Then, the auction can continue with the remaining parties. Therefore, the functionality satisfies the \textit{Robustness}. Moreover, a malicious party cannot delay the execution of the protocol for an arbitrary amount of time, because after timeout, the execution still proceeds. In the best case, when everyone behaves honestly and does not terminate in the middle of the auction process, the computation is processed within $O(1)$ rounds, otherwise, $O(\Delta)$ rounds. Thus, the \textit{Liveness} is satisfied.

In the end, the state channel begins its termination procedure upon receiving a $close()$ message from a party. Next, it awaits obtaining the $close()$ messages from the remaining parties within $1+(|\mathcal{P}| - 1$)$\Delta$ rounds. If all the parties are unanimous in closing the state channel, the functionality returns the deposit back to all parties' account. Otherwise, the state channel remains open. As all parties have to send the $close()$ message to close the state channel, we achieve the consensus on closing.

\subsection{Protocol for double auction state channel}

This section discusses in details the double auction protocol based on state channel that realizes the $\mathcal{F}_{auction}$. The protocol includes two main parts: 1) a Judge contact and 2) Off-chain protocol.

\paragraph{Judge contract.}
The main functionality of this contract is to regulate the state channel and handle disputes.  Every party is able to submit a \textit{state} that everyone has agreed on to this contract. However, the contract only accepts the state with the highest version number. Once a party submits a state $G$, the contract will wait for some deadline $T$ for other parties to raise disputes. See \Cref{fig:judge} for the functionality of the Judge contract. Note that the contract $\mathcal{F}_{Judge}$ has a state variable $channel$ which indicates whether the channel is opened or not. If the channel is not opened ($channel=\bot$), the three functionalities "State submission", "Revocation", and "Close channel" cannot be executed. In the same manner, if the channel is already opened ($channel=created$) then the functionality "Open channel" cannot be executed.

\begin{figure}
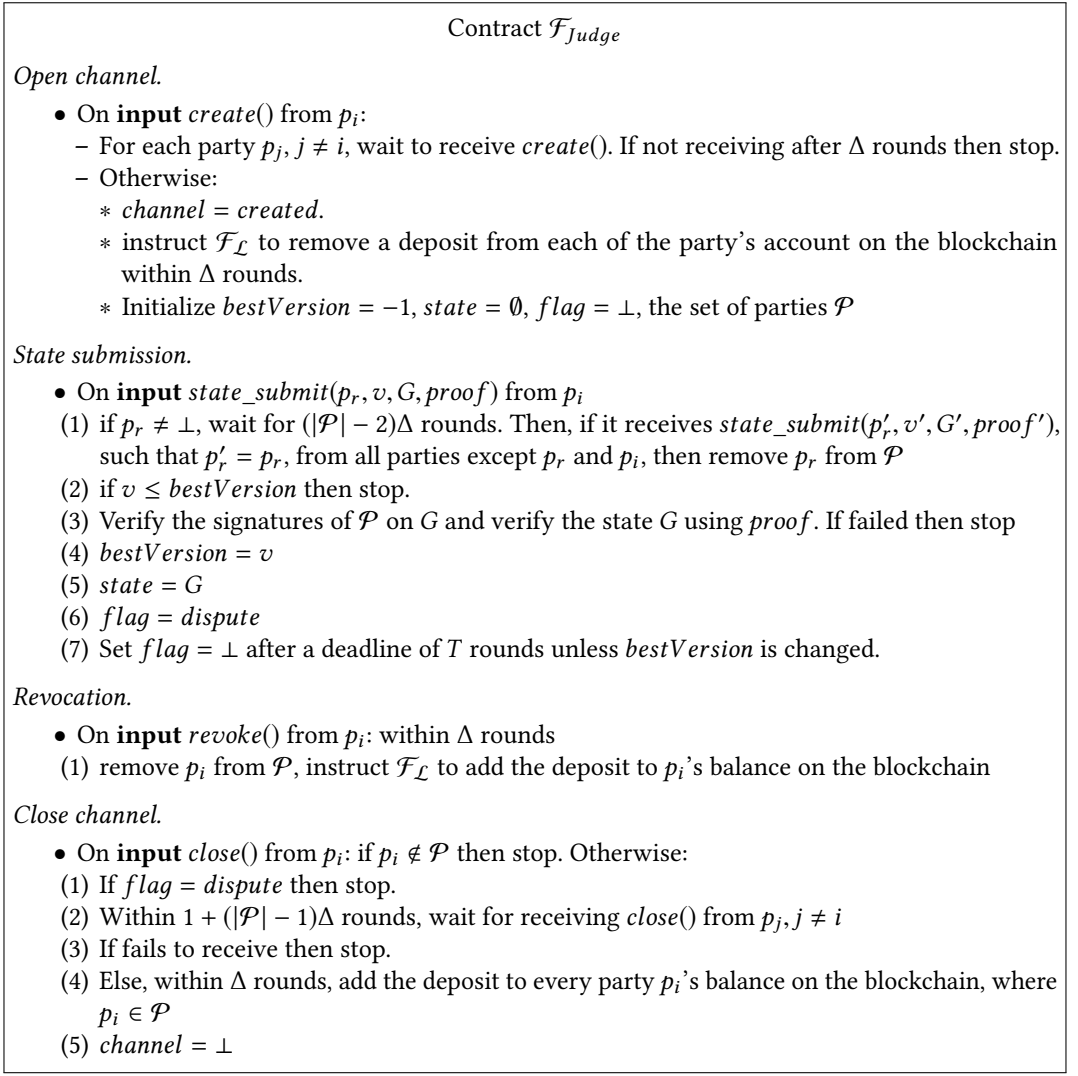

	\noindent\fbox{%
		\parbox{\columnwidth}{%
			\begin{center}
				Contract  $\mathcal{F}_{Judge}$
			\end{center}
			
			\paragraph{Open channel}
			\begin{itemize}
				\item On \textbf{input} $create()$ from $p_i$:
				\begin{itemize}
					\item For each party $p_j$, $j \neq i$, wait to receive $create()$. If not receiving after $\Delta$ rounds then stop.
					\item Otherwise: 
					\begin{itemize}
						\item $channel = created$.
						\item instruct $\mathcal{F}_\mathcal{L}$ to remove a deposit from each of the party's account on the blockchain within $\Delta$ rounds.
						\item Initialize $ bestVersion = -1 $, $state = \emptyset$, $flag = \bot$, the set of parties $\mathcal{P}$
					\end{itemize}
				\end{itemize}
			\end{itemize}

			\paragraph{State submission}
			\begin{itemize}
				\item On \textbf{input} $state\_submit(p_r, v, G, proof)$ from $p_i$
				\begin{enumerate}
					\item if $p_r \neq \bot$, wait for $(|\mathcal{P}|-2)\Delta$ rounds. Then, if it receives $state\_submit(p_r', v', G', proof')$, such that $p_r'=p_r$, from all parties except $p_r$ and $p_i$, then remove $p_r$ from $\mathcal{P}$
					\item if $v \leq bestVersion$ then stop.
					\item Verify the signatures of $\mathcal{P}$ on $G$ and verify the state $G$ using $proof$. If failed then stop
					\item $bestVersion = v$
					\item $state = G$ 
					\item $flag = dispute$
					\item Set $flag = \bot$ after a deadline of $T$ rounds unless $bestVersion$ is changed.
				\end{enumerate}
			\end{itemize}
			
			\paragraph{Revocation}
			\begin{itemize}
				\item On \textbf{input} $revoke()$ from $p_i$: within $\Delta$ rounds 
				\begin{enumerate}
					\item remove $p_i$ from $\mathcal{P}$, instruct $\mathcal{F}_\mathcal{L}$ to add the deposit to $p_i$'s balance on the blockchain
				\end{enumerate}
			\end{itemize}
			
			\paragraph{Close channel}
			\begin{itemize}
				\item On \textbf{input} $close()$ from $p_i$: if $p_i \notin \mathcal{P}$ then stop. Otherwise:
				\begin{enumerate}
					\item If $flag = dispute$ then stop.
					\item Within $1 + (|\mathcal{P}| - 1)\Delta$ rounds, wait for receiving $close()$ from $p_j, j\neq i$
					\item If fails to receive then stop.
					\item Else, within $\Delta$ rounds, add the deposit to every party $p_i$'s balance on the blockchain, where $p_i \in \mathcal{P}$
					\item $channel = \bot$
				\end{enumerate}
			\end{itemize}
		}%
	}
	\caption{Judge contract}
	\label{fig:judge}
\end{figure}

As the contract always maintains the valid state on which all parties have agreed (by verifying all the signatures), we can publicly verify if all parties are following the protocol. A dishonest party who attempts to submit an outdated state will be detected as the smart contract is public, that state would be then overwritten by a more recent state. When the state channel is closed, the contract is now holding the latest state with the final bids of all the parties, and by the immutability of blockchain, no bidder can deny having made the relevant bid. Therefore, this contract satisfies the \textit{Non-Repudiation} and \textit{Public Verifiability} goals.

\paragraph{Off-chain protocol .}
In this section, we present the off-chain protocol  that operates among parties in a double auction process. In the same manner as $\mathcal{F}_{auction}$, the protocol  consists of four parts: (1) Create state channel, (2) Determine best response, (3) Revocation, and (4) Close state channel. \Cref{fig:protocol-overview} illustrates the connections among these parts as well as the execution order of the protocol.

First, to create a new state channel, the environment sends a message $create()$ to one of the parties. Let's denote this initiating party as $p_i$. The detailed protocol is shown in \Cref{fig:protocol1}. $p_i$ will send a $create()$ message to the smart contract $\mathcal{F}_{Judge}$ which will take $\Delta$ rounds to get confirmed on the blockchain. As this message is visible to the whole network, any $p_{j\neq i}$ can detect this event and also send a $create()$ message to $\mathcal{F}_{Judge}$. To detect this event, each $p_j$ needs to retrieve the current state of blockchain which takes 1 round and as there are $n-1$ parties $p_j$, $p_i$ has to wait $1 + (n-1)\Delta$ rounds. If all parties agree on creating the state channel, this process will be successful and the channel will be opened. After that, the smart contract will take a deposit from the account of each party.

\begin{figure}
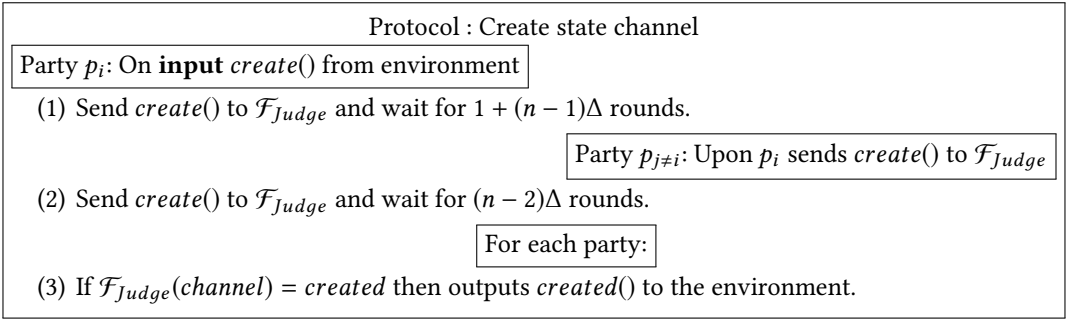

	\noindent\fbox{%
		\parbox{\columnwidth}{%
			\begin{center}
				Protocol  : Create state channel
			\end{center}
			
			\begin{flushleft}
				\fbox{Party $p_i$: On \textbf{input} $create()$ from environment}
			\end{flushleft}
			
			\begin{enumerate}
				\item Send $create()$ to $\mathcal{F}_{Judge}$ and wait for $1 + (n-1)\Delta$ rounds.
				\begin{flushright}
					\fbox{Party $p_{j\neq i}$: Upon $p_i$ sends $create()$ to $\mathcal{F}_{Judge}$}
				\end{flushright}
				\item Send $create()$ to $\mathcal{F}_{Judge}$ and wait for $(n-2)\Delta$ rounds.
				\begin{center}
					\fbox{For each party:}
				\end{center}
				\item If $\mathcal{F}_{Judge}(channel) = created$ then outputs $created()$ to the environment.
			\end{enumerate}
		}%
	}
	\caption{Protocol  : Create state channel}
	\label{fig:protocol1}
\end{figure}

\begin{figure}
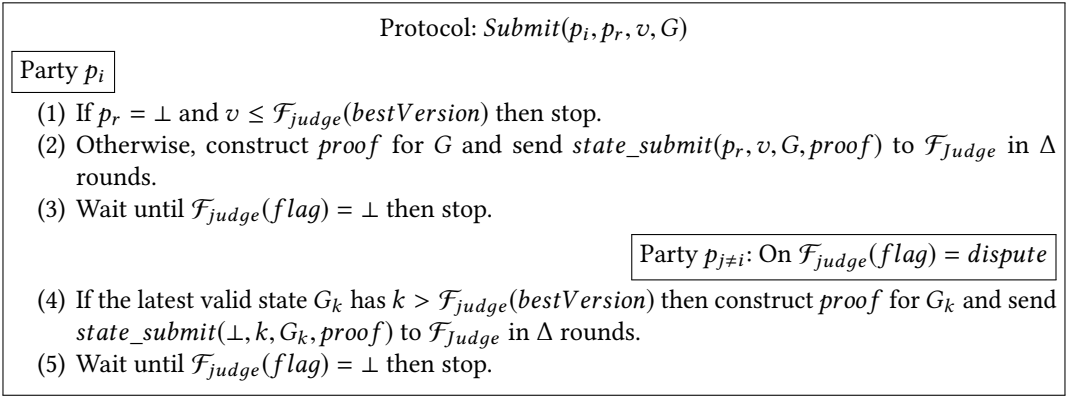

	\noindent\fbox{%
		\parbox{\columnwidth}{%
			\begin{center}
				Protocol: $Submit(p_i, p_r, v, G)$
			\end{center}
			\begin{flushleft}
				\fbox{Party $p_i$}
			\end{flushleft}
			\begin{enumerate}
				\item If $p_r = \bot$ and $v \leq \mathcal{F}_{judge}(bestVersion)$ then stop.
				\item Otherwise, construct $proof$ for $G$ and send $state\_submit(p_r, v, G, proof)$ to $\mathcal{F}_{Judge}$ in $\Delta$ rounds. 
				\item Wait until $\mathcal{F}_{judge}(flag) = \bot$ then stop.
				\begin{flushright}
					\fbox{Party $p_{j\neq i}$: On $\mathcal{F}_{judge}(flag) = dispute$}
				\end{flushright}
				\item If the latest valid state $G_k$ has $k > \mathcal{F}_{judge}(bestVersion)$ then construct $proof$ for $G_k$ and send $state\_submit(\bot, k, G_k, proof)$ to $\mathcal{F}_{Judge}$ in $\Delta$ rounds. 
				\item Wait until $\mathcal{F}_{judge}(flag) = \bot$ then stop.
			\end{enumerate}
		}%
	}
	\caption{Protocol: $Submit$}
	\label{fig:submit}
\end{figure}

When parties run into dispute, they will have to resolve on-chain. In specific, the procedure $Submit()$ as shown in \Cref{fig:submit} allows any party to submit the current state to the smart contract. However, as stated above, $\mathcal{F}_{Judge}$ only considers the valid state that has the highest version number. In this procedure, we also define a $proof$ of a state $G$. Based on the algorithm used for double auction, this $proof$ is anything that can verify whether the calculation of $G$ in an iteration is correct or not. For example, $proof$ can be all the valid bids in that iteration. When any party submits a state, the $\mathcal{F}_{Judge}$ will raise the state variable $flag=dispute$. Upon detecting this event, other parties can submit their states if they have higher version numbers. After a deadline of $T$ rounds, if none of the parties can submit a newer state, $\mathcal{F}_{Judge}$ will set $flag=\bot$ to conclude the dispute period. Furthermore, we also note that this procedure also supports eliminating any dishonest party that does not follow the protocol by setting the parameter $p_r$ to that party. This function requires a unanimous agreement among the remaining honest parties. If a party $p_i$ wants to eliminate a party $p_r$, it will need other parties, except $p_r$, to call the $Submit()$ protocol to remove $p_r$ from $\mathcal{P}$.

\begin{figure}
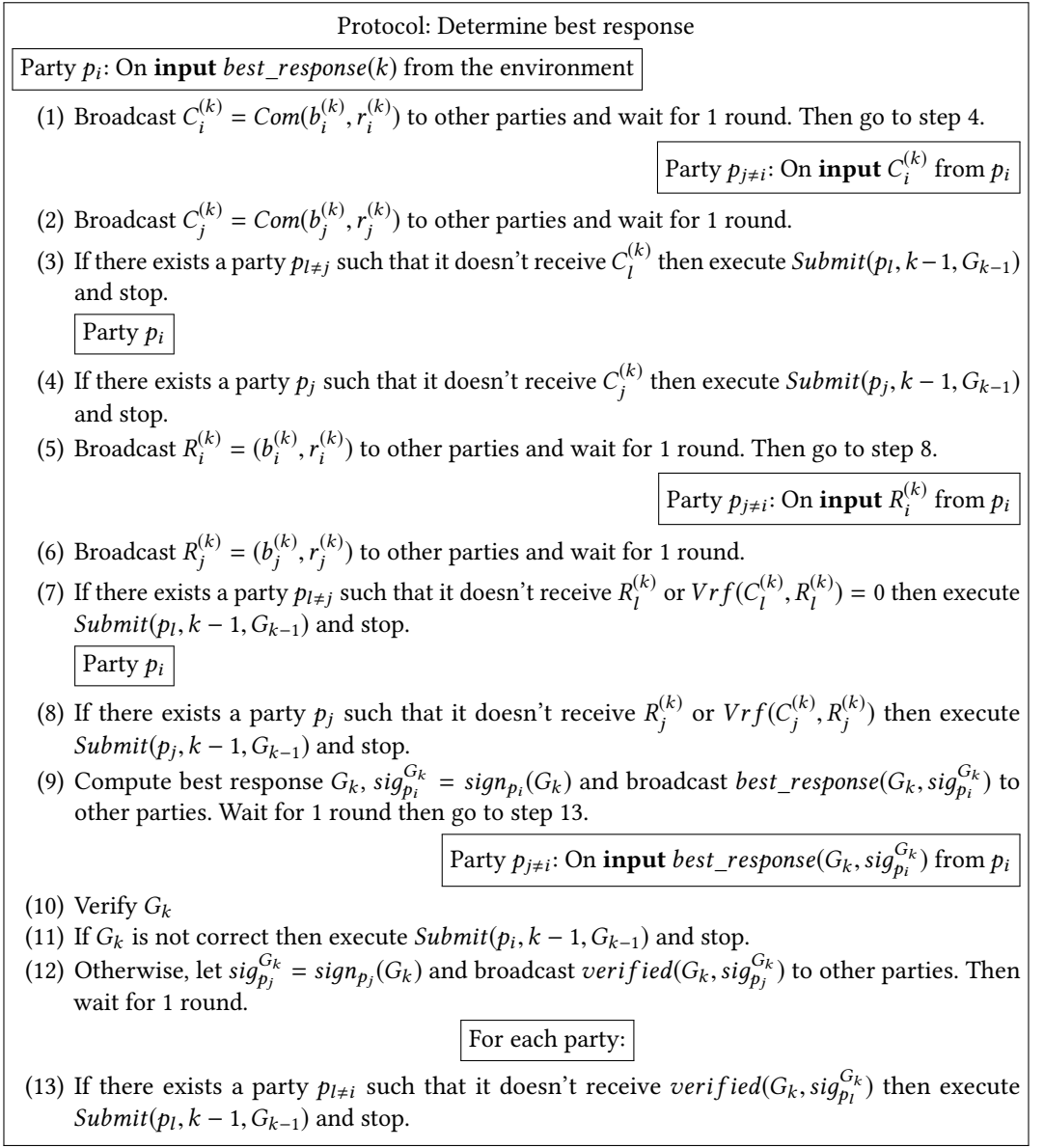

	\noindent\fbox{%
		\parbox{\columnwidth}{%
			\begin{center}
				Protocol: Determine best response
			\end{center}
			\begin{flushleft}
				\fbox{Party $p_i$: On \textbf{input} $best\_response(k)$ from the environment}
			\end{flushleft}
			\begin{enumerate}
				\item Broadcast $C_i^{(k)} = Com(b_i^{(k)}, r_i^{(k)})$ to other parties and wait for 1 round. Then go to step 4.
				\begin{flushright}
					\fbox{Party $p_{j\neq i}$: On \textbf{input} $C_i^{(k)}$ from $p_i$}
				\end{flushright}
				\item Broadcast $C_j^{(k)} = Com(b_j^{(k)}, r_j^{(k)})$ to other parties and wait for 1 round.
				\item If there exists a party $p_{l\neq j}$ such that it doesn't receive $C_l^{(k)}$ then execute $Submit(p_l, k - 1, G_{k-1})$ and stop.
				\begin{flushleft}
					\fbox{Party $p_i$}
				\end{flushleft}
				\item If there exists a party $p_j$ such that it doesn't receive $C_j^{(k)}$ then execute $Submit(p_j, k - 1, G_{k-1})$ and stop.
				\item Broadcast $R_i^{(k)} = (b_i^{(k)}, r_i^{(k)})$ to other parties and wait for 1 round. Then go to step 8.
				\begin{flushright}
					\fbox{Party $p_{j\neq i}$: On \textbf{input} $R_i^{(k)}$ from $p_i$}
				\end{flushright}
				\item Broadcast $R_j^{(k)} = (b_j^{(k)}, r_j^{(k)})$ to other parties and wait for 1 round.
				\item If there exists a party $p_{l\neq j}$ such that it doesn't receive $R_l^{(k)}$ or $Vrf(C_l^{(k)}, R_l^{(k)}) = 0$ then execute $Submit(p_l, k - 1, G_{k-1})$ and stop.
				\begin{flushleft}
					\fbox{Party $p_i$}
				\end{flushleft}
				\item If there exists a party $p_j$ such that it doesn't receive $R_j^{(k)}$ or $Vrf(C_j^{(k)}, R_j^{(k)})$ then execute $Submit(p_j, k - 1, G_{k-1})$ and stop.
				\item Compute best response $G_k$, $sig_{p_i}^{G_k} = sign_{p_i}(G_k)$ and broadcast $best\_response(G_k, sig_{p_i}^{G_k})$ to other parties. Wait for 1 round then go to step 13.

				\begin{flushright}
					\fbox{Party $p_{j\neq i}$: On \textbf{input} $best\_response(G_k, sig_{p_i}^{G_k})$ from $p_i$}
				\end{flushright}
				\item Verify $G_k$
				\item If $G_k$ is not correct then execute $Submit(p_i, k - 1, G_{k-1})$ and stop.
				\item Otherwise, let $sig_{p_j}^{G_k} = sign_{p_j}(G_k)$ and broadcast $verified(G_k, sig_{p_j}^{G_k})$ to other parties. Then wait for 1 round.
				\begin{center}
					\fbox{For each party:}
				\end{center}
				\item If there exists a party $p_{l\neq i}$ such that it doesn't receive $verified(G_k, sig_{p_l}^{G_k})$ then execute $Submit(p_l, k - 1, G_{k-1})$ and stop.
			\end{enumerate}
		}%
	}
	\caption{Protocol  : Determine best response}
	\label{fig:protocol2}
\end{figure}

Next, in \Cref{fig:protocol2}, we present the protocol for determining the best response which only consists of off-chain messages if all the parties are honest. In each iteration $k$, this process starts when the environment sends $best\_response(k)$ to a party $p_i$. Again, this does not violate the trustless property since $p_i$ can be any party chosen at random. First, $p_i$ broadcasts the commitment $C_i^{(k)}$ of its bid which only takes one round since this is an off-chain message. Other parties upon receiving this message will also broadcast their commitments. Then, $p_i$ proceeds to broadcast the opening $R_i^{(k)}$ of its bid and hence, other parties upon receiving this $R_i^{(k)}$ also broadcast their openings. If any party refuses to send their bids or sends an invalid bid, other parties will call the $ Submit $ procedure to eliminate that dishonest party from the auction process. Thus, that party will lose all the deposit. In practice, one may consider refunding a portion of deposit back to that party. To achieve this, we only need to modify the first line of the functionality "State submission" in $\mathcal{F}_{Judge}$ to return a portion of deposit to $p_r$.

During the auction process, some parties may want to abort the auction process. In order to avoid losing the deposit, they must use the Revocation protocol described in \Cref{fig:protocol3} to send a $revoke()$ message to $\mathcal{F}_{Judge}$. In this case, they will get the deposit back in full and be removed from the set $\mathcal{P}$. Other parties upon detecting this operation also update their local $\mathcal{P}$ to ensure the consistency.

\begin{figure}
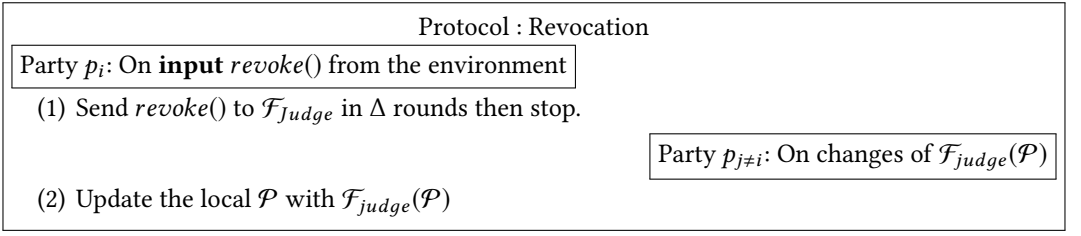

	\noindent\fbox{%
		\parbox{\columnwidth}{%
			\begin{center}
				Protocol  : Revocation
			\end{center}
			\begin{flushleft}
				\fbox{Party $p_i$: On \textbf{input} $revoke()$ from the environment}
			\end{flushleft}
			\begin{enumerate}
				\item Send $revoke()$ to $\mathcal{F}_{Judge}$ in $\Delta$ rounds then stop.
				\begin{flushright}
					\fbox{Party $p_{j\neq i}$: On changes of $\mathcal{F}_{judge}(\mathcal{P})$}
				\end{flushright}
				\item Update the local $\mathcal{P}$ with $\mathcal{F}_{judge}(\mathcal{P})$
			\end{enumerate}
		}%
	}
	\caption{Protocol  : Revocation}
	\label{fig:protocol3}
\end{figure}

Finally, \Cref{fig:close} illustrates the protocol for closing the state channel. One technical point in this protocol is that we must check whether there is any ongoing dispute. If so then we must not close the channel. In the same way of opening the channel, a party $p_i$ also initiates the request by sending a message $close()$ to the smart contract. Upon detecting this event, other parties may also send $close()$. If all parties agreed on closing the channel, they will get the deposit back.

\begin{figure}
	\noindent\fbox{%
		\parbox{\columnwidth}{%
			\begin{center}
				Protocol  : Close state channel
			\end{center}
			
			\begin{flushleft}
				\fbox{Party $p_i$: On \textbf{input} $close()$ from environment}
			\end{flushleft}
			
			\begin{enumerate}
				\item If $\mathcal{F}_{Judge}(flag) = dispute$ then stop.
				\item Send $close()$ to $\mathcal{F}_{Judge}$ and wait for $1 + (|\mathcal{P}|-1)\Delta$ rounds. Then go to step 4
				\begin{flushright}
					\fbox{Party $p_{j\neq i}$: Upon $p_i$ sends $close()$ to $\mathcal{F}_{Judge}$}
				\end{flushright}
				\item Send $close()$ to $\mathcal{F}_{Judge}$ and wait for $(|\mathcal{F}_{auction}(\mathcal{P})|-2)\Delta$ rounds.
				\begin{center}
					\fbox{For each party:}
				\end{center}
				\item Wait for $\Delta$ rounds and check if $\mathcal{F}_{Judge}(channel) = \bot$ then outputs $closed()$ to the environment.
			\end{enumerate}
		}%
	}
	\caption{Protocol  : Close state channel}
	\label{fig:close}
\end{figure}

\begin{figure}
    \centering
    \includegraphics[width=0.85\linewidth]{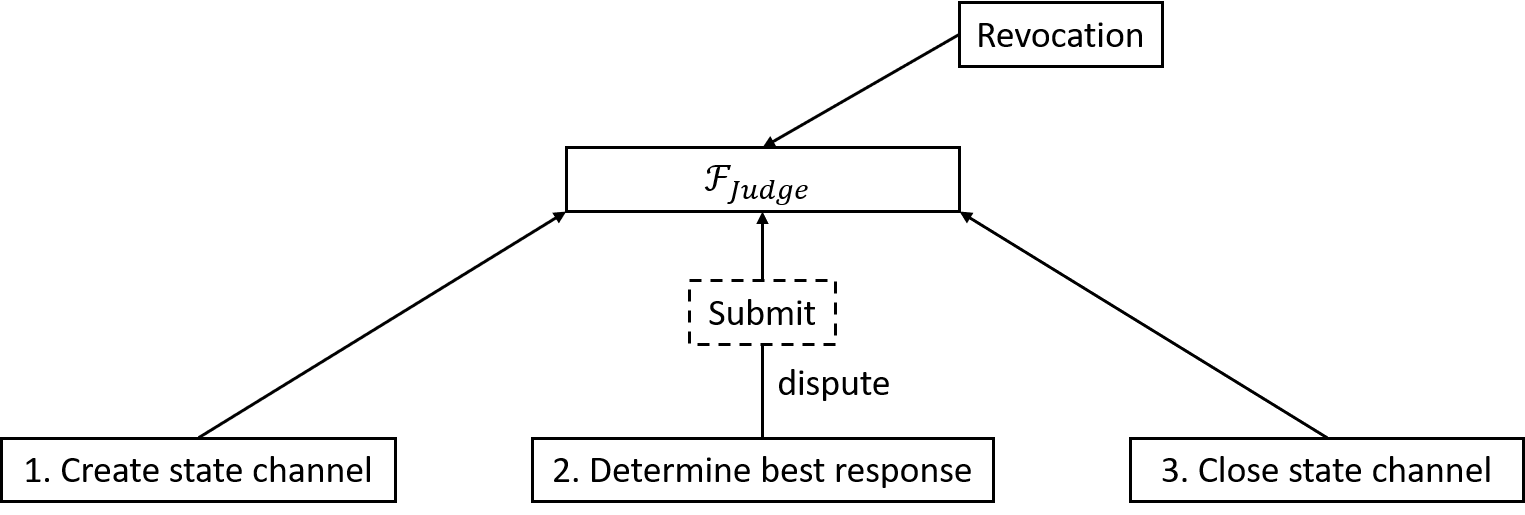}
    \caption{Connections and execution order of the functionalities of the proposed protocol. The "Revocation" functionality can be triggered any time between "Create state channel" and "Close state channel". Any dispute occurred during the "Determine best response" requires a call to the "Submit" functionality.}
    \label{fig:protocol-overview}
\end{figure}

\subsection{Security and Privacy Analysis}
In this section, we prove the security of our solution in the UC model. We denote EXEC$_{\pi, \mathcal{A}, \mathcal{E}}$ as the outputs of the environment $\mathcal{E}$ when interacting with the parties running the protocol $\pi$ and the adversary $\mathcal{A}$. From \cite{malavolta2017concurrency}, we have the following definition:

\begin{definition}[UC-security]
	A protocol $\pi$ UC-realizes an ideal functionality $\mathcal{F}$ if for any adversarial $\mathcal{A}$, there exists a simulator $\mathcal{S}$ such that for any environment $\mathcal{E}$ the outputs EXEC$_{\pi, \mathcal{A}, \mathcal{E}}$ and EXEC$_{\mathcal{F}, \mathcal{S}, \mathcal{E}}$ are computationally indistinguishable
\end{definition}

The main goal of this analysis is to prove that the off-chain protocol UC-realizes the ideal functionality $\mathcal{F}_{auction}$ by constructing a simulator $\mathcal{S}$ in the ideal world that translates every attacker in the real world, such that the two worlds are indistinguishable. To achieve that, we need to ensure the consistency of timings, i. e., the environment $\mathcal{E}$ must receive the same message in the same round in both worlds. Furthermore, in any round, the internal state of each party must be identical between the two worlds, which will make $\mathcal{E}$ unable to perceive whether it is interacting with the real world or the ideal one. 
	
Per Canetti \cite{canetti2001universally}, the strategy for proving the UC-security is constructing the simulator $\mathcal{S}$ that handles the corrupted parties and simulates the $(\mathcal{F}_{judge},\mathcal{F}_\mathcal{L})$-hybrid world while interacting with $\mathcal{F}_{auction}$. The simulator will maintain a copy of the hybrid world internally so that it can turn every behavior in the hybrid world into an indistinguishable one in the real world. We further assume that upon receiving a message from a party, the ideal functionality $\mathcal{F}_{auction}$ will leak that message to the simulator. For simplicity, we do not elaborate these operations when constructing the simulator. Since $\mathcal{S}$ locally runs a copy of the hybrid world, $\mathcal{S}$ knows the behavior of the corrupted parties and the messages sent from $\mathcal{A}$ to $\mathcal{F}_{Judge}$, therefore, $\mathcal{S}$ can instruct the $\mathcal{F}_{auction}$ to update the ledger $\mathcal{L}$ in the same manner as the hybrid word.

In specification of the off-chain protocol, the protocol Submit in \Cref{fig:submit} is called as a sub-routine of the protocol Determine best response in \Cref{fig:protocol2}. Hence, we first define a simulator for the protocol Submit by proving the following lemma:

\begin{lemma}\label{lem:submit}
    Under the assumptions given in \cref{sec:model}, we can construct a simulator for the protocol Submit in the ideal world such that the view of $\mathcal{E}$ remains the same in both the ideal world and the $(\mathcal{F}_{judge},\mathcal{F}_\mathcal{L})$-hybrid world.
\end{lemma}
\begin{proof}
     Let $p_i$ be the party that calls the $Submit()$ protocol, we define $S\_Submit()$ as the simulator of the protocol $Submit()$. If $p_i$ is corrupted, upon $p_i$ sends $state\_submit(p_r, v, G, proof)$ to $\mathcal{F}_{Judge}$
	\begin{enumerate}
	    \item $S\_Submit()$ waits for $\Delta$ rounds. If $p_r = \bot$ then stop.
	    \item Otherwise, $S\_Submit()$ waits for $(|\mathcal{P}| - 2)\Delta$ rounds
		\item If all parties $p_{j\neq \{i,r\}}$ send $state\_submit(p_r, v, G, proof)$ to $\mathcal{F}_{Judge}$ then instruct $\mathcal{F}_{auction}$ to remove $p_r$ from $\mathcal{P}$.
	\end{enumerate}
	If $p_{j\neq i}$ is corrupted, $S\_Submit()$ also updates its $\mathcal{P}$ in the same round as the real world if $\mathcal{F}_{Judge}$  updates $\mathcal{P}$.
	
	Since the protocol $Submit()$ can potentially change the internal state of the $\mathcal{F}_{judge}$ and the parties by removing a party from $\mathcal{P}$, the $S\_Submit()$ ensures that $\mathcal{F}_{auction}$ also performs the same operation in the same round. Hence, the view of both worlds are consistent.
\end{proof}

Finally, we prove the following theorem:

\begin{theorem}
	Under the assumptions given in \cref{sec:model}, the proposed off-chain protocol UC-realizes the ideal functionality $\mathcal{F}_{auction}$ in the $(\mathcal{F}_{judge},\mathcal{F}_\mathcal{L})$-hybrid model.
\end{theorem}
\begin{proof}
	 We provide the description of $\mathcal{S}$ for each of the functionalities as follows.
	
	\paragraph{Open channel.} Let $p_i$ be the party that initiates the request. We inspect the following cases:
	\begin{itemize}
		\item \textit{$p_i$ is corrupted:} 
		Upon $p_i$ sends $create()$ to $\mathcal{F}_{Judge}$
		\begin{enumerate}
			\item $\mathcal{S}$ waits for $\Delta$ rounds
			\item Then sends $create()$ to $\mathcal{F}_{auction}$ to make sure that $\mathcal{F}_{auction}$ receives $create()$ in the same round as $\mathcal{F}_{Judge}$. Then wait for $1+(n-1)\Delta$ rounds
			\item if $\mathcal{F}_{auction}(channel) = created$ then sends $created()$ to $\mathcal{E}$ on behalf of $p_i$.
		\end{enumerate}  
		\item \textit{$p_{j\neq i}$ is corrupted:} Upon $p_i$ sends $create()$ to $\mathcal{F}_{auction}$
		\begin{enumerate}
			\item $\mathcal{S}$ waits for $\Delta$ rounds
			\item If $p_j$ sends $create()$ to $\mathcal{F}_{Judge}$ then $\mathcal{S}$ sends $create()$ to $\mathcal{F}_{auction}$ and wait for $(n-2)\Delta$ rounds
			\item if $\mathcal{F}_{auction}(channel) = created$ then sends $created()$ to $\mathcal{E}$ on behalf of $p_j$
		\end{enumerate}
	\end{itemize}
	
	In all cases above, according to \Cref{fig:protocol1}, $p_i$ or $p_j$ will output $created()$ to $\mathcal{E}$ if $\mathcal{F}_{auction}(channel) = created$. Hence, $\mathcal{S}$ also outputs $created()$ in the same round. Therefore, the environment $\mathcal{E}$ receives the same outputs in the same round in both worlds.
	
	\paragraph{Close channel.}
	Let $p_i$ be the party that initiates the request. We inspect the following cases:
	\begin{itemize}
		\item \textit{$p_i$ is corrupted}: Upon $p_i$ sends $close()$ to $\mathcal{F}_{Judge}$
		\begin{enumerate}
			\item if $\mathcal{F}_{judge}(flag) = dispute$ then stop. Otherwise, $\mathcal{S}$ waits for $\Delta$ rounds
			\item Then sends $close()$ to $\mathcal{F}_{auction}$ to make sure that $\mathcal{F}_{auction}$ receives $close()$ in the same round as $\mathcal{F}_{Judge}$. Then wait for $1+(\mathcal{F}_{auction}(|\mathcal{P}|)-1)\Delta$ rounds
			\item Wait for another $\Delta$ round and check if $\mathcal{F}_{auction}(channel) = \bot$ then sends $close()$ to $\mathcal{E}$ on behalf of $p_i$.
		\end{enumerate} 
		\item \textit{$p_{j\neq i}$ is corrupted:} Upon $p_i$ sends $close()$ to $\mathcal{F}_{auction}$
		\begin{enumerate}
			\item $\mathcal{S}$ waits for $\Delta$ rounds
			\item If $p_j$ sends $close()$ to $\mathcal{F}_{Judge}$ then $\mathcal{S}$ sends $close()$ to $\mathcal{F}_{auction}$ and wait for $(|\mathcal{F}_{auction}(\mathcal{P})|-2)\Delta$ rounds
			\item Wait for another $\Delta$ round and check if $\mathcal{F}_{auction}(channel) = \bot$ then sends $closed()$ to $\mathcal{E}$ on behalf of $p_j$.
		\end{enumerate}
	\end{itemize}
	
	The indistinguishability in the view of $\mathcal{E}$ between the two worlds holds in the same manner as \textit{Open channel}.
	
	\paragraph{Revocation.}
	Let $p_i$ be the party that initiates the request. We inspect the following cases:
	\begin{itemize}
		\item \textit{$p_i$ is corrupted}: Upon $p_i$ sends $revoke()$ to $\mathcal{F}_{Judge}$
		\begin{enumerate}
			\item $\mathcal{S}$ waits for $\Delta$ rounds
			\item Then sends $revoke()$ to $\mathcal{F}_{auction}$ to make sure that $\mathcal{F}_{auction}$ receives $revoke()$ in the same round as $\mathcal{F}_{Judge}$.
		\end{enumerate} 
		\item \textit{$p_{j\neq i}$ is corrupted:} Upon $p_i$ sends $revoke()$ to $\mathcal{F}_{auction}$
		\begin{enumerate}
			\item If $p_j$ updates the local $\mathcal{P}$ then $\mathcal{S}$ also updates its $\mathcal{P}$.
		\end{enumerate}
	\end{itemize}
	
	In both cases, $\mathcal{S}$ ensures that the messages exchanged between the entities are identical in both worlds. Moreover, since $\mathcal{P}$ is updated according to the real world, thus the internal state of the each party are also identical. Therefore, the view of $\mathcal{E}$ between the two worlds are indistinguishable.
	
	\paragraph{Determine best response.}
	Based on \cref{lem:submit}, we define $S\_Submit()$ as the simulator of the protocol $Submit()$ in the ideal world. Let $p_i$ be the party that calculates the best responses. In each iteration $k$, we inspect the following cases:
	\begin{itemize}
		\item \textit{$p_i$ is corrupted:} Upon $p_i$ broadcasts $C_i^{(k)}$ to other parties
		\begin{enumerate}
			\item Send $C_i^{(k)}$ to $\mathcal{F}_{auction}$ and wait for 1 round.
			\item If $\mathcal{F}_{auction}$ removes any party then stop. If $p_i$ executes the $Submit()$ then $\mathcal{S}$ also calls the $S\_Submit()$ in the same round.
			\item Otherwise, if $p_i$ broadcasts $R_i^{(k)}$ to other parties then $\mathcal{S}$ sends $R_i^{(k)}$ to $\mathcal{F}_{auction}$ and waits for 1 round. Else, stop.
			\item If $\mathcal{F}_{auction}$ removes any party then stop. If $p_i$ executes the $Submit()$ then $\mathcal{S}$ also calls the $S\_Submit()$ in the same round. Otherwise, wait for 1 round
			\item Receive $G_k$ from $\mathcal{F}_{auction}$ and wait for 1 round.
			\item If $p_i$ executes the $Submit()$ then $\mathcal{S}$ also calls the $S\_Submit()$ in the same round. Otherwise, stop.
		\end{enumerate} 
		\item \textit{$p_{j \neq i}$ is corrupted:} Upon $p_i$ sends $best\_response(k)$ to $\mathcal{F}_{auction}$
		\begin{enumerate}
			\item Wait until $p_i$ sends $C_i^{(k)}$ to $\mathcal{F}_{auction}$, then forwards that $C_i^{(k)}$ to $p_j$ in the same round.
			\item If $p_j$ broadcasts $C_j^{(k)}$ to other parties then $\mathcal{S}$ sends $C_j^{(k)}$ to $\mathcal{F}_{auction}$. Else, execute $S\_Submit()$ to eliminate the party that made $p_j$ refuse to broadcast and stop.
			\item Wait for 1 round. If $p_i$ sends $R_i^{(k)}$ to $\mathcal{F}_{auction}$, then forwards that $R_i^{(k)}$ to $p_j$ in the same round. Otherwise, stop.
			\item If $p_j$ broadcasts $R_j^{(k)}$ to other parties then $\mathcal{S}$ sends $R_j^{(k)}$ to $\mathcal{F}_{auction}$. Else, execute $S\_Submit()$ to eliminate the party that made $p_j$ refuse to broadcast and stop.
			\item Wait for 1 round, if $\mathcal{S}$ doesn't receive $G_k$ from $\mathcal{F}_{auction}$ then stop. Otherwise, $\mathcal{S}$ forwards that $G_k$ to $p_j$.
			\item If $p_j$ executes the $Submit()$ then $\mathcal{S}$ also calls the $S\_Submit()$ in the same round. Otherwise, stop.
		\end{enumerate}
	\end{itemize}
	
	Since the messages exchanged between any entities are exact in both worlds, the indistinguishability in the view of $\mathcal{E}$ between the two worlds holds.
	
\end{proof}

\section{Implementation and Evaluation} \label{sec:eval}
In this section, we present an evaluation of our proposed double auction framework by running some experiments on a proof-of-concept implementation. Before that, we introduce our novel development framework for distributed computing on state channels.
\subsection{Development framework for distributed computing based on state channels}
To the best of our knowledge, this is the first work that builds a functioning programming framework that can be used to deploy any distributed protocols using blockchain and state channels. We use the Elixir programming language to implement the protocol for the participating parties that run on Erlang virtual machines (EVM). The smart contract for the state channel is implemented using the Solidity programming language, which is the official language for realizing smart contracts on Ethereum. Hence, the whole system can be deployed on an Ethereum blockchain.

The Elixir implementation of the participating parties is designed around the Actor programming model \cite{hewitt1973universal}. This model realizes \textit{actors} as the universal primitive of concurrent computation that use message-passing as the communication mechanism. In response to a message, an actor can make local decisions, send messages to other actors, or dynamically generate more actors. Since actors in Elixir are "location transparent", when sending a message, whether the recipient actor is on the same machine or on another machine, the EVM can manage to deliver the message in both cases. In our framework, each of the parties is treated as a separate actor, and they communicate with each other using asynchronous message-passing mechanism.
Furthermore, in other for the parties to interact with the smart contract, we leverage a remote procedure call protocol encoded in JSON (JSON-RPC) provided by Ethereum. The actors also listen for triggered events (e.g., the channel is opening) from the smart contract and carry out appropriate action.

To develop a distributed computing system using our framework, one would only have to define the followings:
\begin{enumerate}
    \item The state of the system at each iteration.
    \item The operation to be executed at each iteration.
    \item The off-chain messages to be exchanged among the parties of the system.
    \item The rule to resolve disputes.
\end{enumerate}
To demonstrate the feasibility of this development framework, we use it to develop a proof-of-concept implementation of the proposed double auction protocol in the next section.

\subsection{Proof-of-concept implementation of double auction}
We create a proof-of-concept implementation of our proposed double auction protocol using our development framework as illustrated above. Our implementation of the Judge contract closely resembles the protocol from \Cref{fig:judge}, it is developed in the Truffle framework \cite{truffle}, and later deployed on an Ethereum testnet managed by Ganache \cite{ganache}. Our main goal was to illustrate the feasibility and practicality of our off-chain protocols when interacting with the Judge smart contracts. As regards the commitment scheme, we used the SHA-256 hash function with a random string of 64 characters.

One crucial evaluation criteria when working with  smart contracts on the Ethereum blockchain are costs incurred by transaction fees. In Ethereum, these fees are calculated using a special unit called "gas", the fee is paid by the sender of a transaction to the miner that validates that transaction. Transactions in an Ethereum blockchain can provide inputs to and execute smart contract's functions. The amount of gas used for each transaction is determined by the amount of data it sends and the computational effort that it will take to execute certain operations. In addition, depending on the exchange rate between gas and ETH (Ethereum's currency), we can determine the final cost in ETH. In practice, this exchange rate is decided by the person who issues transactions depending on how much they want to prioritize their transactions in the mining pool. As we are using the Ganache testnet, in our calculation, we use the default gas price of 1 gas = $2\times10^{-8}$ Ether.

To deploy the Judge contract, we have to pay 3387400 gas, which corresponds to 0.0677 ETH. Using an exchange rate of 1 ETH = 152.49 USD (as of Nov 2019), the deployment of the contract costs approximately 10.3 USD. We also note that our proof-of-concept implementation did not aim to optimize the deployment cost since we implement every functionality on a single smart contract. To mitigate this cost, we would only need to transfer all functionalities of the Judge contract into an external library. Hence, this library only requires one-time deployment, instances of the Judge contract can refer to this library and re-use its functionalities. Therefore, we can save much of the deployment cost since we don't have to re-deploy all the functionalities with each instance of the Judge contract. Thus, in the following section, we disregard this deployment cost when evaluating the performance of our system.

\subsection{Experiments}
For the experiments, we run the system on machines equipped with an Intel Core i7-8550U CPU 1.8 GHz and 16 GB of RAM. The mean latency for transmitting a 32-byte message from one party to another is about 101.2 ms. We assume that, in the beginning, the parties of the system collected the public keys of each other. Moreover, each party also created an account on the Ethereum testnet. Since this is a one-time setup procedure, we do not take it into consideration when evaluating our system. To receive events from the Judge contract, each party periodically checks for triggered events from the Judge contract every 100 ms. To illustrate the double auction process, we follow the setup of a numerical experiment in \cite{zou2016efficient} that consists of 10 parties in total, in which 6 are buyers and 4 are sellers. The auction process takes a total of 300 iterations to reach NE. 

First, we implemented the straw-man design as described in \Cref{sec:strawman}. \Cref{fig:sm} shows the cumulative gas consumption over time. As can be seen, the gas consumption increases linearly with the iterations. This happens because each iteration requires parties to send on-chain transactions to the smart contract. In sum, the straw-man design used 121913080 gas, which corresponds to 2.44 ETH. This means the design incurs a cost of about 372 USD to carry out the double auction process. Furthermore, if we consider the Ethereum's block time of 15 seconds per block, it could take over 75 minutes to complete the process (assuming that the waiting time in the mining pool is negligible).

\begin{figure}
	\centering
	\subfloat[Straw-man design]{
	    \includegraphics[width=0.45\linewidth]{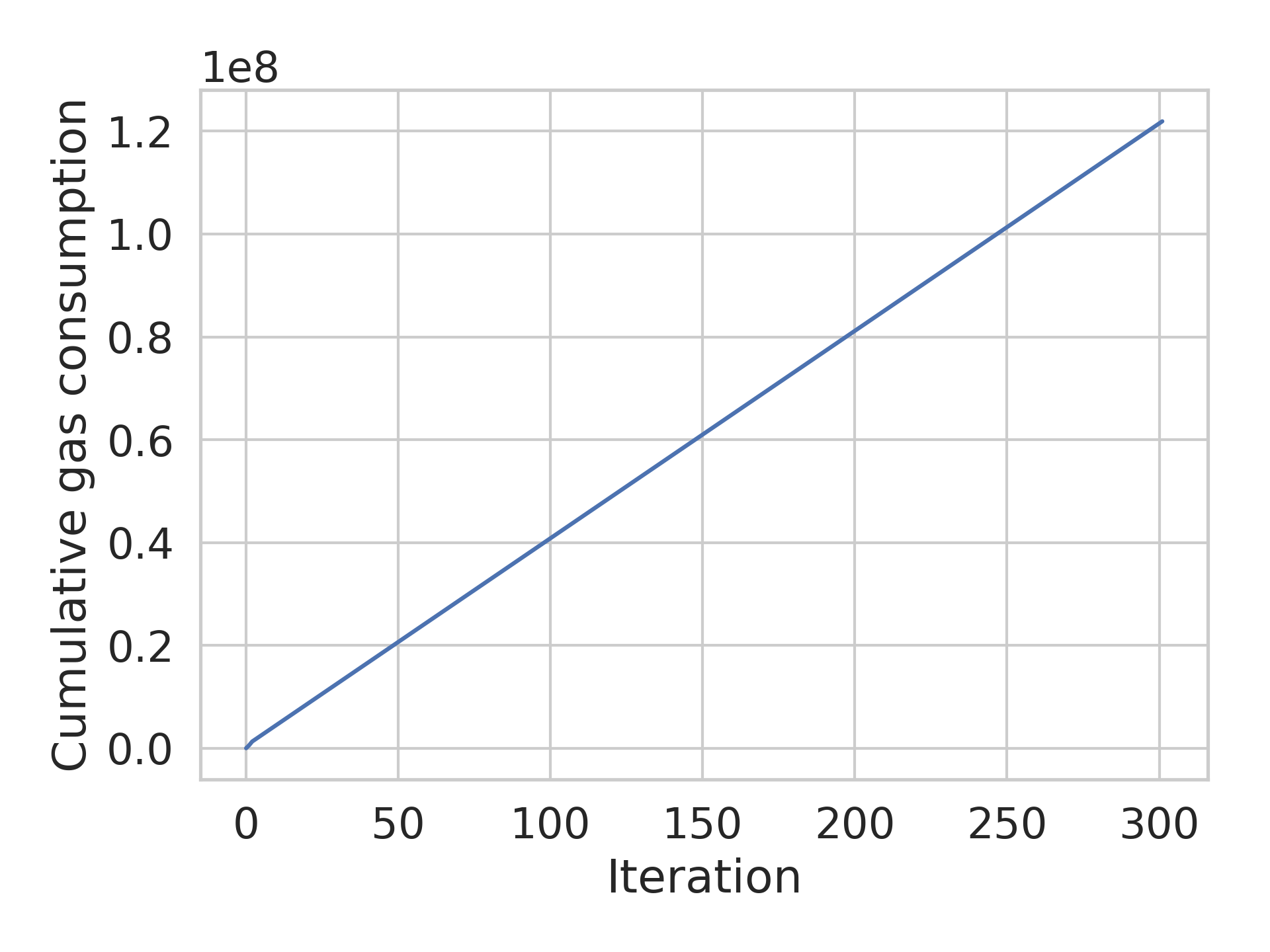}
	    \label{fig:sm}
	}	
	\hfill
	\subfloat[Using state channel]{
	    \includegraphics[width=0.45\linewidth]{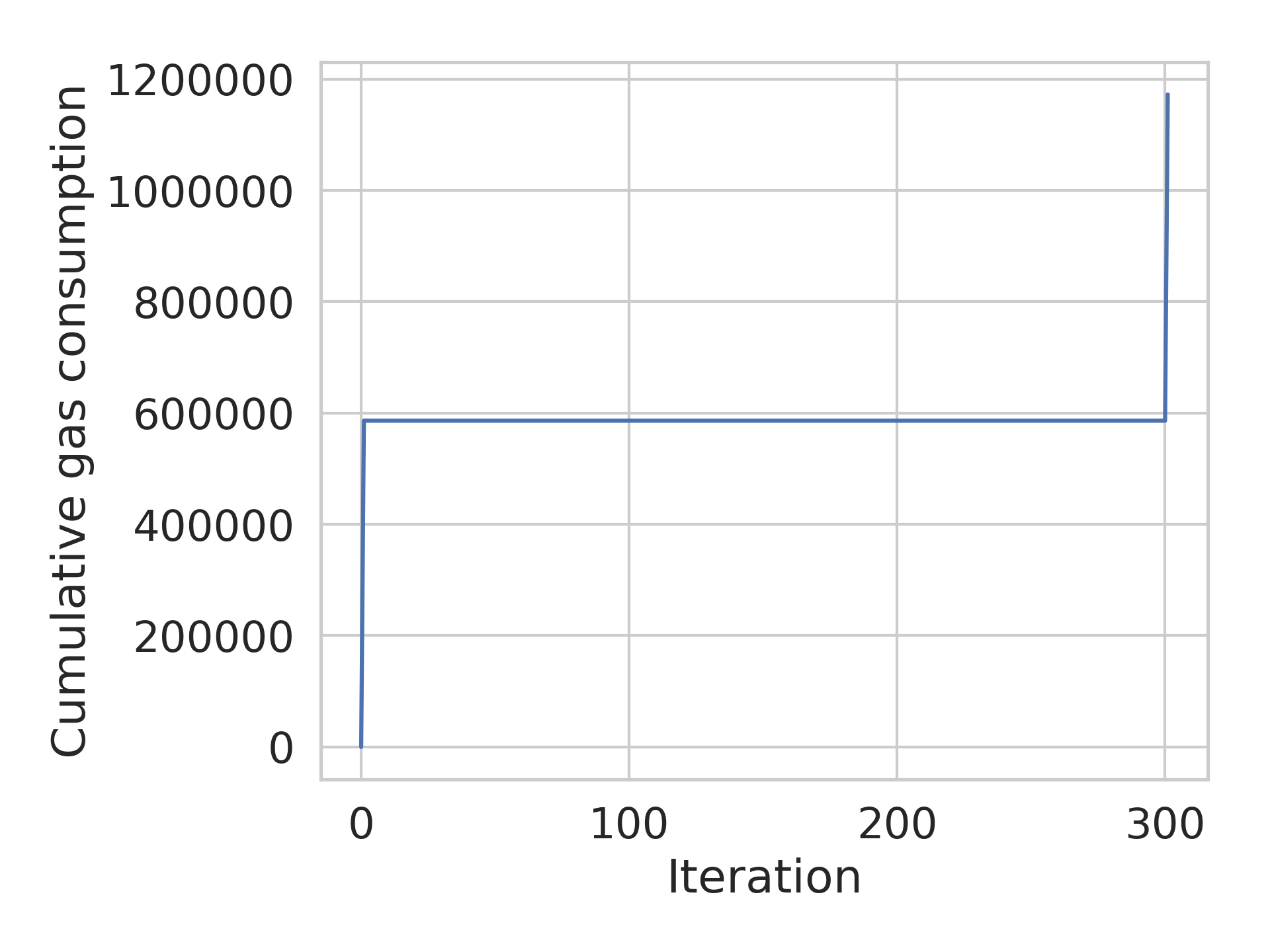}
    	\label{fig:sc}
	}
	\caption{Cumulative gas consumption over time}
\end{figure}

Next, we run the double auction on our proof-of-concept implementation. If all the parties unanimously agree to create the state channel, they need to issue 10 on-chain transactions that cost about 0.0117 ETH (each party needs to send 1 transaction to the smart contract). For closing the state channel, they need another 10 transactions that cost about 0.0109 ETH. The whole double auction process is executed off-chain and doesn't cost any gas or ETH. In \Cref{fig:sc}, which shows the cumulative gas consumption over time, we can observe that only the first and last iteration incurs some amount of gas to create and close the state channel, respectively. As a result, the proposed protocol only incurs 20 on-chain transactions and 1133420 gas. With comparison to the straw-man design, we reduced the amount of transactions and gas by 99\%. Additionally, regarding the nominal cost, our solution requires only 0.0226 ETH, which is about 3 USD, to conduct the auction process. \Cref{tab:res} summarizes the costs of execution of our proposed system.

\begin{table}[]
	\caption{The estimated costs for executing the transactions as well as the message complexity of the proposed double auction protocol when running with 10 parties.}
	\label{tab:res}
	\begin{tabular}{@{}lcccc@{}}
		\toprule
		& \# on-chain tx & Gas    &  ETH    & \# off-chain message \\ \midrule
		Create state channel  & 10                       & 586180 & 0.0117 & 0                    \\
		Double auction process     & 0                        & 0      & 0      & 81000                 \\
		Close state channel & 10                       & 547240 & 0.0109 & 0                    \\ \bottomrule
	\end{tabular}
\end{table}

For a more detailed analysis, we look into other scenarios and measure some aspects of the system. One scenario to consider is when a party becomes dishonest and attempts to close the state channel with an outdated state. In order to do that, the dishonest party has to submit the outdated state using the $state\_submit$ function of the Judge contract, this costs about 52845 gas. After that, another party would submit a more recent state that can overwrite that outdated state, this costs another 52845 gas. Consequently, the cost for settling a disagreement would be about 0.0021 ETH for both the dishonest and honest party.

Another scenario that is worth looking into is the dynamic leave of parties. The two cases supported by our system are (1) eliminating dishonest parties and (2) voluntarily aborting the auction process. In the former case, when eliminating a party via the $state\_submit$ function, it requires 9 on-chain transactions that cost a total of 0.0122 ETH. In the latter case, using the Revocation protocol, only the leaving party needs to issue a transaction that costs about 0.0011 ETH.

To analyze the execution cost, we focus on the cost of computing digital signatures and message complexity since these are additional work that the parties have to perform apart from computing best responses. As in the protocol design, in each iteration, each party has to sign the current state and broadcast it to all other parties. A state of an iteration is defined as the list of parties' best responses and it is about 80 bytes (each party's bid profile is 8 bytes). The time it takes to sign the state is about 1.4 ms and to verify the signature is about 0.8 ms. As for broadcasting bids among parties, a commitment of 32 bytes (using SHA-256 hash function) is sent and later accompanied by the opening of 72 bytes (including the bid and the random string). This incurs a total of 270 off-chain messages that are transmitted per iteration with the total size of 23 KB. If we consider the whole auction process that involves 300 iterations, the total size of the messages transmitted among the parties is only 6.9 MB. This emphasizes the practicality of our proposed solution because of low transmission overhead.

Finally, we measure the end-to-end latency of the auction process. For the creating and closing state channel, each will take one block time (i.e., about 15 seconds in Ethereum). The 300 iterations of determining best responses takes only about 8.3 minutes. To test the scalability, we observe how the system latency changes when we increase the number of parties. For a fair comparison, we fix the number of iterations for the double auction process to 300. In \Cref{fig:scale}, we show the end-to-end latency with 10, 40, 70, and 100 parties. The total time to create and close the state channel remains the same because each operation only needs one block time as a block in Ethereum can store about 380 transactions.  When we increase the number of parties to 40, the system needs 833.4 seconds to finish the auction process. The rise in the latency comes from the fact that the size of the state at each iteration increases (from 80 bytes to 320 bytes), hence, it takes longer to transmit the state as well as to compute the digital signatures. However, we note that with 40 parties, this latency is only 1.6 times as much as the latency with 10 parties. Likewise, when we increase the number of parties by 9 times (to 100 parties), the latency only increases by 2.9 times. Therefore, we can see that the proposed solution is able to scale with the number of parties.

To better capture the performance of our system under heavy load, \Cref{fig:scale1} illustrates the end-to-end delay with thousands of trading parties over 300 iterations. For creating and closing state channel, it would take 6, 12, 16, 22, and 28 block time with 1000, 2000, 3000, 4000, and 5000 parties, respectively. As can be seen, the latency increases linearly with the number of parties. Under this load, the size of the state reaches 40000 bytes at 5000 parties, and the end-to-end latency is dominated by the time it takes to transmit the state at each iteration. Therefore, the performance of our system is only limited by the underlying communication channel.

\begin{figure}
    \centering
    \subfloat[]{\includegraphics[width=0.5\linewidth]{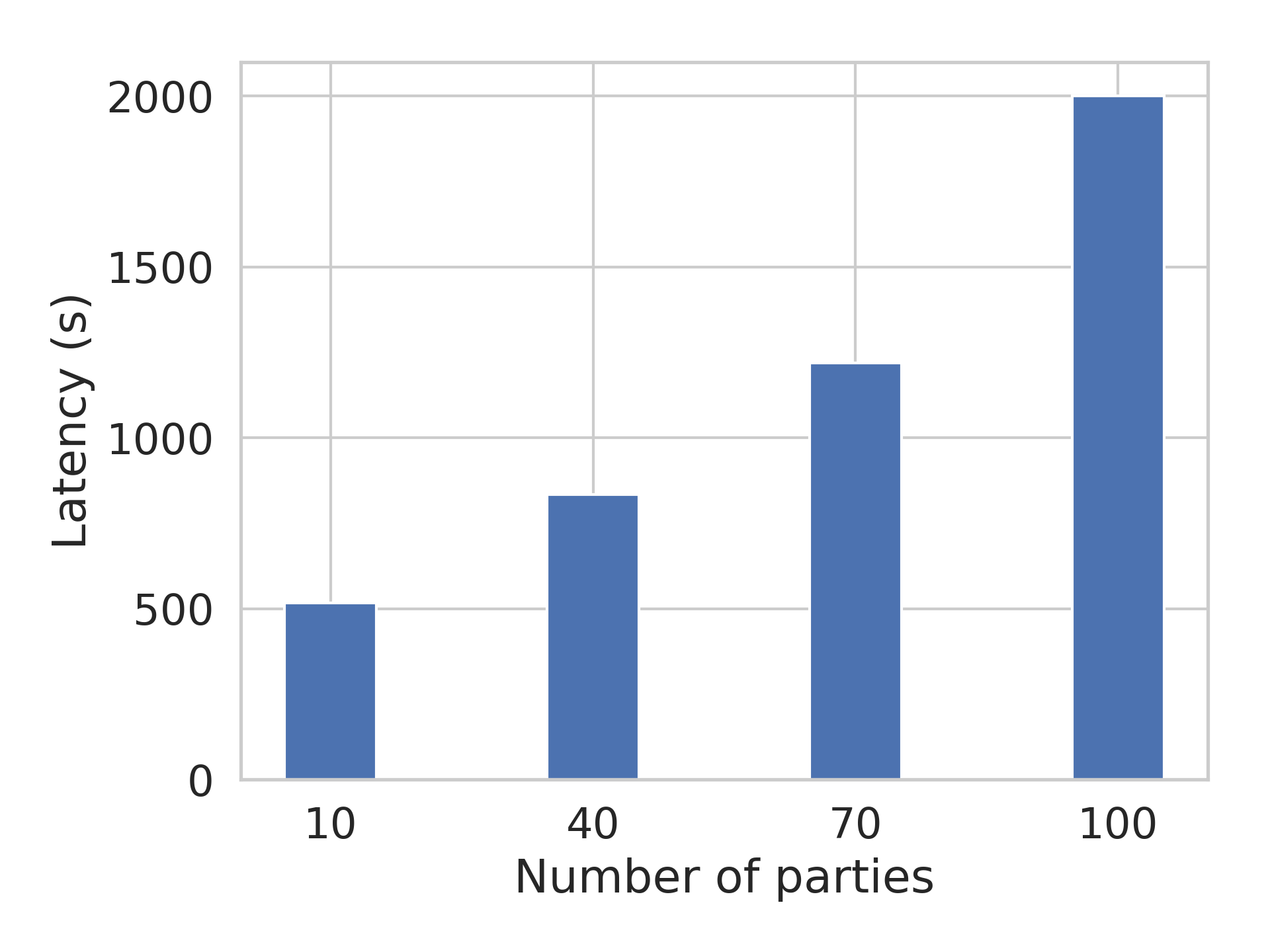}
    \label{fig:scale}
    }
    \subfloat[]{\includegraphics[width=0.5\linewidth]{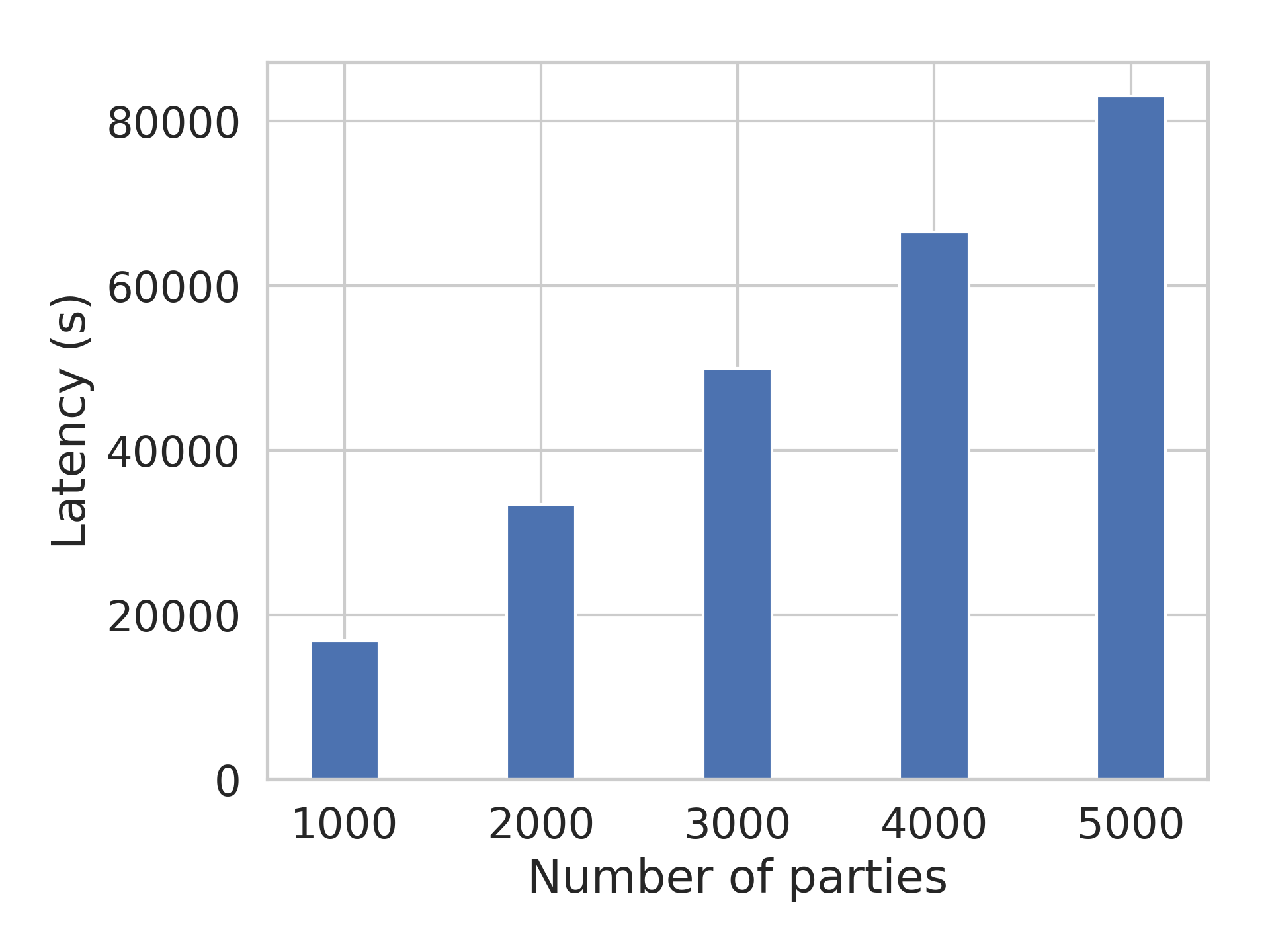}
    \label{fig:scale1}
    }
    \caption{End-to-end latency}
    
\end{figure}

In summary, the proof-of-concept implementation has demonstrated that our solution is both feasible and practical, all the functionalities work according to the design in \Cref{sec:state}. With comparison to the straw-man design, our implementation has resulted in a gigantic saving of money and time. Moreover, as the proposed solution induces a relatively small overhead, it is able support a high number of trading parties. 

\section{Conclusion}\label{sec:conclude}
In this paper, we have proposed a novel framework based on blockchain that enables a complete decentralized and trustless iterative double auction. That is, all parties can participate in the auction process without having to rely on an auctioneer and they do not have to trust one another. With an extension of the state channel technology, we were able to specify a protocol that reduces the blockchain transactions to avoid high transaction fee and latency. We have provided a formal specification of the framework and our protocol was proven to be secured in the UC model. Finally, we have developed a proof-of-concept implementation and perform experiments to validate the feasibility and practicality of our solution.


\bibliographystyle{ACM-Reference-Format}
\bibliography{bibliography}


\end{document}